\documentclass[runningheads,envcountsame]{llncs}
\usepackage{amssymb}
\usepackage{amsmath}
\usepackage{comment}
\usepackage[disable]{todonotes}
\usepackage[utf8]{inputenc}
\usepackage{inputenc}
\usepackage{shuffle}
\usepackage{multirow}
\usepackage{listings}
\usepackage{mathabx}
\usepackage{hyperref}

\usepackage{shuffle}

\usepackage{tikz}
\usetikzlibrary{positioning,shadows,arrows}
\usetikzlibrary{arrows,automata,positioning,calc}

\usepackage{diagbox}
\usepackage{slashbox}
\usepackage{tikz}
\usetikzlibrary{matrix}

\newcommand{\perm}{\operatorname{perm}}

\usepackage{blindtext,tikz}
\usetikzlibrary{calc}

\usepackage{apxproof}
\theoremstyle{plain}
\newtheoremrep{theorem}{Theorem}
\newtheoremrep{proposition}[theorem]{Proposition}
\newtheoremrep{lemma}[theorem]{Lemma}
\newtheoremrep{claim}[theorem]{Claim}
\newtheoremrep{conjecture}[theorem]{Conjecture}
\newtheoremrep{corollary}[theorem]{Corollary}
\theoremstyle{definition}
\newtheoremrep{definition}[theorem]{Definition}
 
\usepackage{fancyhdr}

\DeclareFontFamily{U}{bigshuffle}{}
\DeclareFontShape{U}{bigshuffle}{m}{n}{
  <5-8> s*[1.7] shuffle7
  <8->  s*[1.7] shuffle10
}{}
\DeclareSymbolFont{BigShuffle}{U}{bigshuffle}{m}{n}
\DeclareMathSymbol\bigshuffle{\mathop}{BigShuffle}{"001}
\DeclareMathSymbol\bigcshuffle{\mathop}{BigShuffle}{"002}

\newcommand{\suff}{\operatorname{Suff}}
\newcommand{\factor}{\operatorname{Fact}}

\newcommand{\pref}{\operatorname{Pref}}

\newcommand{\NL}{\textsf{NL}}
\newcommand{\NP}{\textsf{NP}}
\newcommand{\PSPACE}{\textsf{PSPACE}}

\newcommand{\PTIME}{\textsf{P}}

\usepackage{tabularx,lipsum,environ,amsmath,amssymb}

\newcounter{problemcounter}
\makeatletter
\newcommand{\problemtitle}[1]{\gdef\@problemtitle{#1}}
\newcommand{\probleminput}[1]{\gdef\@probleminput{#1}}
\newcommand{\problemquestion}[1]{\gdef\@problemquestion{#1}}
\NewEnviron{decproblem}{
  \refstepcounter{problemcounter}
  \problemtitle{}\probleminput{}\problemquestion{}
  \BODY
  \par\addvspace{.5\baselineskip}
  \noindent
  \begin{tabularx}{\textwidth}{@{\hspace{\parindent}} l X c}
    \multicolumn{2}{@{\hspace{\parindent}}l}{\textbf{Decision Problem \theproblemcounter:} \@problemtitle} \\
    \textbf{Input:} & \@probleminput \\
    \textbf{Question:} & \@problemquestion
  \end{tabularx}
  \par\addvspace{.5\baselineskip}
}
\makeatother

 \newenvironment{myclaiminproof}[1]{\medskip\par\noindent\underline{Claim:}\space#1}{}
 \newenvironment{myclaimproof}[1]{\begin{quote}\par\noindent\emph{Proof of the Claim:}\space#1}{[\emph{End, Proof of the Claim}]\end{quote}}

\fancypagestyle{AllPages}{
\chead{2016 IEEE/ACM International Conference on Advances in Social Networks Analysis and Mining (ASONAM)}
}
\fancypagestyle{FirstPage}{
\chead{2016 IEEE/ACM International Conference on Advances in Social Networks Analysis and Mining (ASONAM)}
\lfoot{IEEE/ACM ASONAM 2016, August 18-21\\2016, San Francisco, CA,     USA\\
978-1-5090-2846-7/16/\$~\copyright~2016 IEEE}
}

\begin{document}

\title{Computational Complexity of Synchronization under Sparse Regular Constraints}
\titlerunning{Synchronization under Sparse Regular Constraints}

%
%
\author{Stefan Hoffmann\orcidID{0000-0002-7866-075X}}
\authorrunning{S. Hoffmann}
%
\institute{Informatikwissenschaften, FB IV, 
  Universit\"at Trier, Germany, 
  \email{hoffmanns@informatik.uni-trier.de}}
\maketitle              
\begin{abstract}
 The constrained synchronization problem (CSP) asks
 for a synchronizing word of a given input automaton
 contained in a regular set of constraints. It could be viewed
 as a special case of synchronization of a discrete event system
 under supervisory control.
 Here, we study the computational complexity of this 
 problem for the class of sparse regular constraint languages.
 We give a new characterization of sparse regular sets, which
 equal the bounded regular sets,
 and derive a full classification of the computational complexity
 of CSP for letter-bounded regular constraint
 languages, which properly
 contain the strictly bounded regular languages.
 Then, we introduce strongly self-synchronizing codes
 and investigate CSP for  bounded languages induced by these codes.
 With our previous result, we deduce a full classification
 for these languages as well.
 In both cases, depending on the constraint language, our problem
 becomes $\NP$-complete or polynomial time solvable.
 
 %
\end{abstract}
%
%
%




\section{Introduction} 
\label{sec:introduction}

A deterministic semi-automaton is \emph{synchronizing} if it admits a reset word, i.e., a word which leads to a definite
state, regardless of the starting state. This notion has a wide range of applications, from software testing, circuit synthesis, communication engineering and the like, see~\cite{DBLP:journals/et/ChoJSP93,San2005,Vol2008}.  
The famous \v{C}ern\'y conjecture \cite{Cer64}
states that a minimal synchronizing word, for an $n$ state automaton, has length
at most $(n-1)^2$. 
We refer to the mentioned survey articles for details~\cite{San2005,Vol2008}.

Due to its importance, the notion of synchronization has undergone a range of generalizations and variations
for other automata models.
The paper~\cite{DBLP:conf/mfcs/FernauGHHVW19} introduced the constrained synchronization problem (CSP\footnote{In computer science the
acronym CSP is usually used for the constraint satisfaction problem~\cite{Lecoutre09}. However, as here we are not concerned with
constrained satisfaction problems at all, no confusion should arise.}). 
In this problem, we search for a synchronizing word coming from a specific subset of allowed
input sequences. 
To sketch a few applications:

%


\begin{description}

\item[Reset State.] 

In~\cite{DBLP:conf/mfcs/FernauGHHVW19} one motivating example was the demand that a system, or automaton thereof, to synchronize has to first enter a ``directing'' mode, perform a sequence of
operations, and then has to leave this operating mode and enter the ``normal
operating mode'' again. In the most simple case, this constraint
could be modelled by $ab^*a$, which, as it turns out~\cite{DBLP:conf/mfcs/FernauGHHVW19},
yields an \NP-complete CSP.
Even more generally, it might be possible that a system -- a remotely controlled
rover on a distant planet, a satellite in orbit, or a lost autonomous vehicle
-- is not allowed to execute all commands in every possible
order, but  certain commands are only allowed in certain order or after
other commands have been executed. All of this imposes constraints
on the possible reset sequences.

\item[Part Orienters.] Suppose parts arrive at a manufacturing site and they
need to be sorted and oriented before assembly. Practical considerations
favor methods which require little or no sensing, employ simple devices,
and are as robust as possible. This could be achieved as follows.
We put parts to be oriented on a conveyor belt which takes them to the assembly
point and let the stream of the parts encounter a series of passive obstacles placed along
the belt. Much research on synchronizing automata was motivated
by this application~\cite{DBLP:journals/algorithmica/ChenI95,DBLP:journals/siamcomp/Eppstein90,DBLP:journals/trob/ErdmannM88,DBLP:journals/algorithmica/Goldberg93,DBLP:conf/focs/Natarajan86,DBLP:journals/ijrr/Natarajan89,Vol2008}
and I refer to~\cite{Vol2008} for
an illustrative example. 
Now, furthermore, assume the passive components could not be placed at random along
the belt, but have to obey some restrictions, or restrictions
in what order they are allowed to happen. These could be due
to the availability of components, requirements how to lay things out
or physical restrictions. 

\item[Supervisory Control.] The CSP
 could also be viewed of as 
 supervisory control
 of a discrete event system (DES) that is given by an automaton
 and whose event sequence is modelled by a formal language~\cite{DBLP:books/daglib/0034521,RamadgeWonham87,wonham2019}.
 In this framework, a DES has a set of controllable
 and uncontrollable events.
 Dependent
 on the event sequence that occurred so far,
 the supervisor is able to restrict the set of events
 that are possible in the next step, where, however,
 he can only limit the use of controllable events.
 So, if we want to (globally) reset a finite state DES~\cite{Alves2020} under supervisory control, 
 this is equivalent to CSP.

\item[Biocomputing.] In~\cite{Benenson2003,Benenson2001} DNA molecules have been used as both
hardware and software for finite automata of nanoscaling size, see also~\cite{Vol2008}. 
For instance, Benenson et al~\cite{Benenson2003} produced ``a `soup of automata', that is, a solution
containing $3 \times 10^{12}$ identical automata per $\mu1$. All these molecular
automata can work in parallel on different inputs, thus ending up in different
and unpredictable states. In contrast to an electronic computer, one cannot
reset such a system by just pressing a button; instead, in order to synchronously bring
each automaton to its start state, one should spice the soup with (sufficiently
many copies of) a DNA molecule whose nucleotide sequences encodes a reset word''~\cite{Vol2008}.
Now, it might be possible that certain sequences, or subsequences,
are not possible as they might have unwanted biological side-effects, 
or might destroy the molecules at all.

\item[Reduction Procedure.] 
 This example is more formal and comes from attempts to solve the \v{C}ern\'y conjecture~\cite{Vol2008}.
 In~\cite{Gusev2012} a special rank factorization~\cite{piziak99}
 for automata was introduced from which
 smaller automata could be derived.
 Then, it was shown that the original automaton
 is synchronizing if and only if the reduced automaton
 admits a synchronizing word in a certain regular constraint language,
 and the reset threshold, i.e, the lengths of the shortest
 synchronizing word, of the original automaton
 could be bounded by that of the shortest one in the constraint
 language for the reduced automaton.
 
\end{description}


%
%
%

In~\cite{DBLP:conf/mfcs/FernauGHHVW19}, a complete analysis of the complexity landscape when the constraint language is given by small partial automata was done. It is natural to extend this result to other language classes.

In general there exist constraint languages yielding 
\PSPACE-complete constrained problems~\cite{DBLP:conf/mfcs/FernauGHHVW19}. 
A language is polycyclic~\cite{DBLP:conf/ictcs/Hoffmann20}, if it is recognizable by an automaton
such that every strongly connected component forms a single cycle,
and a language is sparse~\cite{DBLP:reference/hfl/Yu97} if only polynomially many words of a specific length are in the language.
As shown in~\cite{DBLP:conf/ictcs/Hoffmann20}
for polycyclic languages, which, as we show, equal the sparse regular
languages, the problem is always in \NP. This motivates investigating this class further. 
Also, as written in more detail in Remark~\ref{rem:motivation_strictly_bounded},
a subclass of these languages has a close relation to the commutative languages, and as for commutative constraint
languages a trichotomy result has been established~\cite{DBLP:conf/cocoon/Hoffmann20},
tackling the sparse languages seems to be the next logical step.
In fact, we show a dichotomy result for a subclass
that contains the class corresponding to the commutative languages.
Additionally, as has been noted in~\cite{DBLP:conf/mfcs/FernauGHHVW19},
the constraint language $ab^*a$ is the smallest language, in terms of a recognizing
automaton, giving an \NP-complete CSP. The class of languages
for which our dichotomy holds true contains this language.



Let us mention that restricting the solution space by a regular language
has also been applied in other areas, for example to topological sorting~\cite{DBLP:conf/icalp/AmarilliP18},
solving word equations~\cite{Diekert98TR,DBLP:journals/iandc/DiekertGH05}, constraint programming~\cite{DBLP:conf/cp/Pesant04}, or
shortest path problems~\cite{DBLP:journals/ipl/Romeuf88}.
The famous road coloring theorem~\cite{adler1970similarity,Trahtman09} states
that every finite strongly connected and directed aperiodic graph of uniform out-degree admits a labelling of its edges 
such that a synchronizing
automaton results. A related problem to our problem of constrained synchronization is to restrict the possible labelling(s), and
this problem was investigated in~\cite{DBLP:journals/jcss/VorelR19}.

\paragraph{Outline and Contribution.} 
Here, we look at the complexity landscape for sparse regular constraint
languages.
In Section~\ref{sec:sparse}
 we introduce the sparse languages and show that the regular sparse
 languages are characterized by
 polycyclic automata introduced in~\cite{DBLP:conf/ictcs/Hoffmann20}.
 A similar characterization in terms of non-deterministic
 automata was already given in~\cite[Lemma 2]{DBLP:journals/ijfcs/GawrychowskiKRS10}.
 In this sense, we extend this characterization to the deterministic
 case. As for polycyclic constraint automata the constrained
 problem is always in \NP, see~\cite[Theorem 2]{DBLP:conf/ictcs/Hoffmann20},
 we can deduce the same for sparse regular constraint languages, which
 equal the bounded regular languages~\cite{DBLP:journals/eik/LatteuxT84}.

In Section~\ref{sec:strictly_bounded_case}
we introduce the letter-bounded languages, a proper subset of the sparse languages,
and show that for letter-bounded  constraint languages, the constrained
synchronization problem is either in \PTIME\ or \NP-complete.

 The difficulty why we cannot handle the general case yet lies in the fact that
 in the reductions, 
 in the general case, we need auxiliary states and it is not clear
 how to handle them properly, i.e, 
 how to synchronize them properly while staying inside the constraint language.

 In Section~\ref{sec:strongly_self_sync}
 we introduce the class of strongly self-synchronizing
 codes. 
 The strongly self-synchronizing codes allow us to handle these auxiliary states mentioned before.
 We show that for homomorphisms given by such codes,
 the constrained problem for the homomorphic image of a language has the same computational
 complexity as for the original language.
 This result holds in general, and hence is of independent interest.
 Here we apply it to the special case
 of bounded, or sparse, regular languages given by such codes.

 Lastly, we present a bounded language giving an \NP-complete constrained
 problem that could not be handled by our methods so far.

\section{Preliminaries and Definitions}
\label{sec:preliminaries}



We assume the reader to have some basic knowledge in computational complexity theory and formal language theory, as contained, e.g., in~\cite{HopUll79}. For instance, we make use of  regular expressions to describe languages.
By $\Sigma$ we denote the \emph{alphabet}, a finite set.
For a word $w \in \Sigma^*$ we denote by $|w|$ its \emph{length},
and, for a symbol $x \in \Sigma$, we write $|w|_x$ to denote the \emph{number of occurrences of $x$}
in the word. We denote the empty word, i.e., the word of length zero, by $\varepsilon$.
For $L \subseteq \Sigma^*$, we set $\pref(L) = \{ u \in \Sigma^* \mid \exists v \in \Sigma^* : uv \in L \}$.
A word $u \in \Sigma^*$ is a \emph{factor} (or \emph{infix}) of $w \in \Sigma^*$ if there exist words $x,y \in \Sigma^*$ such that $w = xuy$. 
For $U, V \subseteq \Sigma^*$, we set $U\cdot V = UV = \{ uv \mid u \in U, v \in V \}$
and 
$U^0 = \{ \varepsilon \}$, $U^{i+1} = U^i U$, 
and $U^* = \bigcup_{i \ge 0} U^i$ and $U^+ = \bigcup_{i > 0} U^i$.
We also make use of complexity classes like $\PTIME$, $\NP$, or $\PSPACE$.
With  $\leq^{\log}_m$ we denote a logspace many-one reduction.
If for two problems $L_1, L_2$ it holds that $L_1 \leq^{\log}_m L_2$ and $L_2 \leq^{\log}_m L_1$, then we write $L_1 \equiv^{\log}_m L_2$.

A \emph{partial deterministic finite automaton (PDFA)} is a tuple $\mathcal A = (\Sigma, Q, \delta, q_0, F)$,
where $\Sigma$ is a finite set of \emph{input symbols},~$Q$ is the finite \emph{state set}, $q_0 \in Q$ the \emph{start state}, $F \subseteq Q$ the \emph{final state set} and $\delta \colon Q\times \Sigma \rightharpoonup Q$ the \emph{partial transition function}.
The \emph{partial transition function} $\delta \colon Q\times \Sigma \rightharpoonup Q$ extends to words from $\Sigma^*$ in the usual way. 
Furthermore, for $S \subseteq Q$ and $w \in \Sigma^*$, we set $\delta(S, w) = \{\,\delta(q, w) \mid \mbox{$\delta(q,w)$ is defined and } q \in S\,\}$.
We call $\mathcal A$ \emph{complete} if~$\delta$ is defined for every $(q,a)\in Q \times \Sigma$.
If $|\Sigma| = 1$, we call $\mathcal A$ a \emph{unary automaton} and
 $L \subseteq \Sigma^*$ is also called a \emph{unary language}.
The set $L(\mathcal A) = \{\, w \in \Sigma^* \mid \delta(q_0, w) \in F\,\}$ denotes the language
\emph{recognized} 
by~$\mathcal A$.

A \emph{deterministic and complete semi-automaton (DCSA)} $\mathcal A = (\Sigma, Q, \delta)$
is a deterministic and complete finite automaton without a specified start state
and with no specified set of final states.
When the context is clear, we call both deterministic finite automata and semi-automata simply \emph{automata}.


A complete automaton $\mathcal A$ is called \emph{synchronizing} if there exists a word $w \in \Sigma^*$ with $|\delta(Q, w)| = 1$. In this case, we call $w$ a \emph{synchronizing word} for $\mathcal A$.
We call a state $q\in Q$ with $\delta(Q, w)=\{q\}$ for some $w\in \Sigma^*$ a \emph{synchronizing state}.
For a semi-automaton (or PDFA) with state set $Q$ and transition function $\delta : Q \times \Sigma \rightharpoonup Q$,
a state $q$ is called a \emph{sink state}, if for all $x \in \Sigma$ we have $\delta(q,x) = q$.
Note that, if a synchronizing automaton has a sink state, then the
synchronizing state is unique and must equal the sink state.

In~\cite{DBLP:conf/mfcs/FernauGHHVW19} the \emph{constrained synchronization problem (CSP)} 
was defined for a fixed PDFA
$\mathcal B = (\Sigma, P, \mu, p_0, F)$. 

\begin{decproblem}\label{def:problem_L-constr_Sync}
  \problemtitle{\cite{DBLP:conf/mfcs/FernauGHHVW19}~\textsc{$L(\mathcal B)$-Constr-Sync}}
  \probleminput{DCSA $\mathcal A = (\Sigma, Q, \delta)$.}
  \problemquestion{Is there a synchronizing word $w \in \Sigma^*$ for $\mathcal A$ with  $w \in L(\mathcal B)$?}
\end{decproblem}

The automaton $\mathcal B$ will be called the \emph{constraint automaton}.
If an automaton~$\mathcal A$ is a yes-instance of \textsc{$L(\mathcal B)$-Constr-Sync} we call $\mathcal A$ \emph{synchronizing with respect to~$\mathcal{B}$}. 
Occasionally,
we do not specify $\mathcal{B}$ and rather talk about \textsc{$L$-Constr-Sync}.
For example, for the unconstrained case, 
we have $\Sigma^*\textsc{-Constr-Sync}\in \PTIME$~\cite{Cer64,Vol2008}.

\begin{toappendix}
The following obvious remark, stating that the set of synchronizing words
is a two-sided ideal, will be used frequently without further mentioning.

\begin{lemma}
	\label{lem:append_sync} 
	Let $\mathcal A = (\Sigma, Q, \delta)$ be a deterministic and complete semi-automaton and $w\in \Sigma^*$ be a synchronizing word for $\mathcal A$. Then for every $u, v \in \Sigma^*$, the word $uwv$ is also synchronizing. 
\end{lemma}
\end{toappendix}

In our $\NP$-hardness reduction, we will need the following problem
from~\cite{DBLP:conf/ictcs/Hoffmann20}.

\begin{decproblem}\label{def:unary_set_transpoer}
  \problemtitle{\textsc{DisjointSetTransporter}}
  \probleminput{DCSA $\mathcal A = (\Sigma, Q, \delta)$ and disjoint $S, T \subseteq Q$.}
  \problemquestion{Is there a word $w \in \Sigma^*$ such that $\delta(S, w) \subseteq T$?}
\end{decproblem}


\begin{theorem}
\label{prop:set_transporter_np_complete}
 For unary deterministic and complete input semi-automata the problem \textsc{DisjointSetTransporter}
 is $\NP$-complete.
\end{theorem}

A PDFA $\mathcal A=(\Sigma, Q, \delta, q_0, F)$
is called \emph{polycyclic}, if for each $q \in Q$
there exists $u \in \Sigma^*$ such that
$\{ w \in \Sigma^* \mid \delta(q, w) = q \} \subseteq u^*$. 
A PDFA is polycyclic if and only if every strongly connected
component consists of a single cycle~\cite[Proposition 3]{DBLP:conf/ictcs/Hoffmann20}, where each transition in the cycle is labelled by precisely one
letter. 
Formally, 
for each strongly connected component $S \subseteq Q$ 
and $q \in S$, we
have\footnote{In~\cite{DBLP:conf/ictcs/Hoffmann20}, I made an error in my formalization
by writing
$|\{ \delta(q, x) : x \in \Sigma, \delta(q, x) \mbox{ is defined } \} \cap S| \le 1$.} 
$|\{ x : x \in \Sigma \mbox{ and } \delta(q, x) \mbox{ is defined and in $S$} \}| \le 1$ (note that in the special case $|S| = 1$, the aforementioned set might be empty if the single state in $S$ has no self-loops).
A precursor of this characterization of polycyclic automata in
a special case was given in~\cite{DBLP:conf/stacs/GanardiHKLM18}
under the term \emph{linear cycle automata}.

\begin{toappendix}

\begin{example} A few examples and non-examples of polycyclic automata 
(start and final states not indicated, as they are irrelevant for these examples).
 
  \begin{tikzpicture}[node distance=15mm, auto]
  \tikzstyle{vertex}=[circle,fill=black!25,minimum size=10pt,inner sep=0pt]
   
   \node at (0.5,0) {};
  
   \foreach \name/\x in {s/1.5, 2/2.25}
    \node[vertex] (G-\name) at (\x,0) {};
    
   \foreach \name/\angle/\text in {P-1/180/5, P-2/108/6, 
                                   P-3/36/7, P-4/-36/8, P-5/-108/9}
    \node[vertex,xshift=4cm] (\name) at (\angle:0.9cm) {};

  \foreach \name/\angle/\text in {Q-1/90/5, Q-2/180/6, 
                                  Q-3/0/7}
    \node[vertex,xshift=2cm,yshift=-1.9cm] (\name) at (\angle:0.6cm) {};
    
  \node[vertex] (S1) at (5.7,0.5) {};
  \node[vertex] (S2) at (6.5,0.5) {};
  
  \foreach \from/\to/\label in {G-s/G-2/a,G-2/P-1/b,P-1/P-2/a,P-2/P-3/c,P-3/P-4/c,P-4/P-5/c,P-5/P-1/a,P-5/Q-1/b,Q-1/Q-3/a,Q-3/Q-2/b,Q-2/Q-1/b}
     \path[->] (\from) edge node {$\label$} (\to);
     
  \path[->] (P-3) edge node {$a$} (S1);
  \path[->] (S1) edge node {$b$} (S2);
  \path[->] (S2) edge [loop below] node {$a$} (S2);
   \node at (5,-1.7) {polycyclic PDFA};
  
   \node[vertex] (S1) at (9,0.5) {};
   \path[->] (S1) edge [loop above] node {\emph{$a,b$}} (S1);
   \node at (9,0) {not polycyclic};

    \foreach \name/\angle/\text in {T-1/90/5, T-2/180/6, 
                                  T-3/0/7}
    \node[vertex,xshift=9cm,yshift=-1.6cm] (\name) at (\angle:0.85cm) {};
    
    \path[->] (T-1) edge node {\emph{$a,b$}} (T-3)
              (T-3) edge node [above]  {$a$} (T-2)
              (T-2) edge node  {$a$} (T-1);
    
      \node at (9,-2) {not polycyclic};
 \end{tikzpicture} 
 \end{example}    

Note that a complete DFA over a non-unary alphabet is never polycyclic.
This could be seen with Theorem~\ref{thm:bounded_characterization},
as the complement of a language with polynomial growth cannot have polynomial growth
itself, or by using the property that every strongly connected
component is a single cycle. 
Then, for two distinct letters and a strongly connected component, if one maps a state in the cycle into the cycle, the other one must leave the cycle. However, if we topologically sort
the strongly connected components, the component at the end must be closed
under any letter.
\end{toappendix}

\begin{toappendix}

We will need the following result from~\cite{DBLP:conf/cocoon/Hoffmann20}.

\begin{lemma}[\cite{DBLP:conf/cocoon/Hoffmann20}]
	\label{lem:union}
	Let $\mathcal X$ denote any of the complexity classes
	$\PTIME$, $\NP$ or $\PSPACE$.
	If $L(\mathcal B)$ is a finite union of languages $L(\mathcal B_1),
	L(\mathcal B_2), \dots, L(\mathcal B_n)$ such that for each $1\leq i\leq n$
the problem $L(\mathcal B_i)\textsc{-Constr-Sync}\in \mathcal X$, 
	then we have $L(\mathcal B)\textsc{-Constr-Sync } \in \mathcal X$.
\end{lemma}
\end{toappendix}

The following slight generalization  
of~\cite[Theorem 27]{DBLP:conf/mfcs/FernauGHHVW19} will be needed.

\begin{propositionrep}
\label{prop:hom_lower_bound_complexity}
 Let $\varphi \colon \Sigma^* \to \Gamma^*$ be a homomorphism.
 Then, for each regular $L \subseteq \Sigma^*$, we have 
 $\varphi(L)\textsc{-Constr-Sync} \le_m^{\log} L\textsc{-Constr-Sync}$.
\end{propositionrep}
\begin{proof}
 
 

 Let $\mathcal A = (\Gamma, Q, \delta)$ be a DCSA.
	We want to know if it is synchronizing with respect to~$\varphi(L)$.
	Build the automaton $\mathcal A' = (\Sigma, Q, \delta')$ according to the rule
	$$
	\delta'(p, x) = q\quad\mbox{if and only if}\quad
        \delta(p, \varphi(x)) = q,$$ for $x\in\Sigma^*$.
	As~$\varphi$ is a mapping,~$\mathcal A'$ is indeed deterministic and
        complete, as~$\mathcal A$ is a DCSA. As the homomorphism~$\varphi$ is
        independent of $\mathcal A$, automaton~$\mathcal A'$ can be constructed from~$\mathcal A$
        in logarithmic space.  Next we prove that the 
        translation is indeed a reduction.

        If $u \in \varphi(L)$ is some synchronizing word for~$\mathcal A$, then there is
        some~$s\in Q$ such that $\delta(r,u)=s$, for all $r\in Q$.  By
        choice of $u$, we find $w \in L$ such that $u = \varphi(w)$. As with
        $\delta(r,\varphi(w))=s$, it follows $\delta'(r,w)=s$, hence~$w$
        is a synchronizing word for~$\mathcal A$.
        Conversely, if $w \in L$ is a synchronizing word
        for~$\mathcal A'$, then there is some $s\in Q$ such that
        $\delta'(r,w)=s$, for all $r\in Q$.  Further, $\varphi(w)$ is
        a synchronizing word for $\mathcal A$, as by definition for all $r\in
        Q$, we have $\delta(r,\varphi(w))=s$. 
\end{proof}

\begin{toappendix}

In the proofs of this appendix, we will need the following results
and constructions.

If $|L(\mathcal B)| = 1$, then $L(\mathcal B)\textsc{-Constr-Sync}$
is obviously in $\PTIME$. Simply feed this single word into the input
semi-automaton for every state and check if a unique state results.
Hence by Lemma \ref{lem:union} the next is implied.

\begin{lemma}\label{lem:finite} 
 Let $\mathcal B = (\Sigma, P, \mu, p_0, F)$ be a constraint automaton
 such that $L(\mathcal B)$ is finite, then
 $L(\mathcal B)\textsc{-Constr-Sync} \in \PTIME$.
\end{lemma}

\begin{lemma}\label{lem:union_single_final_state}
 Every regular language could be written as a finite union of regular languages
 recognizable by automata with a single final state.
\end{lemma}
\begin{proof}
 Let $\mathcal A = (\Sigma, Q, \delta, q_0, F)$ be a DFA.
 Then, $L(\mathcal A) = \bigcup_{q \in F} \{ w \in \Sigma^* \mid \delta(q_0, w) = q\}$
 and $\{ w \in \Sigma^* \mid \delta(q_0, w) = q\} = L((\Sigma, Q, \delta, q_0, \{q\}))$.
\end{proof}

Let $\Sigma = \{a\}$ be a unary alphabet. Suppose $L \subseteq \Sigma^{\ast}$ is regular
with a recognizing complete deterministic automaton $\mathcal A = (\Sigma, Q, \delta, q_0, F)$. Then, by considering
the sequence of states $\delta(q_0, a^1), \delta(q_0, a^2), \delta(q_0, a^3), \ldots$, we find numbers\footnote{Note that these numbers are independent of the language recognized by the automaton.}
$i \ge 0, p > 0$ with $i+p$ minimal such that $\delta(q_0, a^i) = \delta(q_0, a^{i+p})$.
We call these numbers the \emph{index} $i$ and the \emph{period} $p$ of the automaton $\mathcal A$.
If $Q = \{\delta(q_0, a^m) \mid m \ge 0 \}$, i.e., every state is reachable from the start state, then  $i + p = |Q|$.
 In our discussion, unary languages that are recognized by
 automata with a single final state appear.
 
 \begin{lemma}[\cite{DBLP:conf/cai/Hoffmann19}]
\label{lem::unary_single_final}
  Let $L \subseteq \{a\}^{\ast}$ be a unary language that is recognized
  by an automaton with a single final state, index $i$ and period $p$.
  Then either $L = \{u\}$ with $|u| < i$ (and if the automaton is minimal we would have $p = 1$),
  or $L$ is infinite with $L = a^{i+m}(a^p)^{\ast}$ and $0 \le m < p$. Hence
  two words $u,v$ with $\min\{|u|, |v|\} \ge i$ are both in $L$ or not if and only
  if $|u| \equiv |v| \pmod{p}$.
 \end{lemma}
 
We need an additional construction, which we call \emph{inflating}
a given automaton $\mathcal A = (\Sigma, Q, \delta, q_0, F)$

\begin{definition}[Aut. Inflation by a factor $N > 0$]
\label{def:inflate_aut}
 Let $\mathcal A = (\Sigma, Q, \delta, q_0, F)$
 be a given automaton and $N > 0$.
 Then, the 
 \emph{inflated automaton (of $\mathcal A$) by $N$ } 
 is $\mathcal A' = (\Sigma, Q', \delta', q_0, F)$, where
\[
 Q' = Q \cup \bigcup_{x\in \Sigma} ( Q_{1,x} \cup \ldots \cup Q_{N-1, x} )
\]
and the $Q_{i,x}$ are disjoint copies of $Q$.
The states $Q_{1,x}, \ldots, Q_{N-1,x}$
are called \emph{auxiliary states} in this context.
Then, for $q \in Q'$ and $x \in \Sigma$, set
\[
 \delta'(q, x) = \left\{
 \begin{array}{ll}
  q_{i+1,x}    & \mbox{if } q = q_{i,x} \land i \in \{1,\ldots, N-2\}; \\
  \delta(q, x) & \mbox{if } q = q_{N-1,x}; \\
  q_{1,x}      & \mbox{if } q \in Q; \\
  q            & \mbox{otherwise.}
 \end{array}
 \right.
\]
\end{definition}

Intuitively, a transition labelled by $x$ is replaced by a path labelled by $x^N$.
Note that, for $q,q' \in Q$
\[
 \delta(q, x) = q' \mbox{ in } \mathcal A
 \Leftrightarrow \delta'(q, x^N) = q' \mbox{ in } \mathcal A'
\]
and $\delta'(q, w) \in Q$ implies that $|w|$
is divisible by $N$.
\end{toappendix}

\section{Sparse and Bounded Regular Languages}
\label{sec:sparse}

Here, in Theorem~\ref{thm:sparse_in_NP}, we establish that for constraint languages from the class of sparse regular languages, which equals the class of the bounded regular languages~\cite{DBLP:journals/eik/LatteuxT84}, the constrained problem
is always in \NP.

A language $L \subseteq \Sigma^*$ is \emph{sparse},
if there exists $c \ge 0$
such that, for every $n \ge 0$, we have
$L \cap \Sigma^n \in O(n^c)$.
Sparse languages were introduced into computational complexity
theory by Berman \& Hartmanis~\cite{DBLP:journals/siamcomp/BermanH77}.
Later, it was established by Mahaney that if there exists
a sparse \NP-complete set (under polynomial-time many-one reductions),
then $\PTIME = \NP$~\cite{DBLP:journals/jcss/Mahaney82}.
For a survey on the relevance of sparse sets in computational complexity theory, see~\cite{DBLP:conf/mfcs/HartmanisM80}.

A language $L \subseteq \Sigma^*$ is called \emph{bounded},
if there exist $w_1, \ldots, w_k \in \Sigma^*$
such that $L \subseteq w_1^* \ldots w_k^*$. 
Bounded languages were introduced by Ginsburg \& Spanier~\cite{GinsburgSpanier64}.

We will need the following representation of the bounded regular languages.


\begin{theorem}[\cite{GinsburgSpanier66}]
\label{thm:bounded_regular_form}
 A language $L \subseteq w_1^* \cdots w_k^*$ is regular if and only if
 it is a finite union of languages of the form $L_1 \cdots L_k$, where each $L_i \subseteq w_i^*$ is regular.
\end{theorem}

It is known that the class of sparse regular languages equals
the class of bounded regular languages~\cite{DBLP:journals/eik/LatteuxT84},
or see~\cite{Pin2020,DBLP:reference/hfl/Yu97}, where the bounded languages are not mentioned
but the equivalence is implied by their results and Theorem~\ref{thm:bounded_regular_form}.
The next results links this class to the polycylic PDFAs.

\begin{propositionrep}
\label{thm:bounded_characterization}
 Let $L \subseteq \Sigma^*$ be regular. Then, $L$ is sparse
 if and only if it is recognizable by a polycyclic PDFA.
\end{propositionrep}
\begin{proof}
 In~\cite{DBLP:journals/eik/LatteuxT84} is was shown that the context-free sparse languages are precisely the context-free bounded languages, which
 gives our first two equivalences.
 A result from~\cite[Lemma 2]{DBLP:journals/ijfcs/GawrychowskiKRS10}
 readily implies that if a language is recognized by a polycyclic PDFA, then
 it must be sparse.
 Lastly, we show that every bounded regular language
 is recognizable by a polycyclic automaton, which finishes the proof.
 
 \medskip 
 
 \noindent\underline{Claim:} For $w \in \Sigma^*$.
  Then, any regular $L \subseteq w^*$
  is recognizable by a polycyclic PDFA.
 \begin{quote}
     \emph{Proof of the Claim.}
      Let $w \in \Sigma^*$ and $L \subseteq w^*$ be a regular language.
 If $w = \varepsilon$, then $L = \{\varepsilon\}$, which is obviously recognizable
 by a polycyclic automaton. So, suppose $|w| > 0$.
%
%
 Let $a$ be an arbitrary symbol
 and define a homomorphism $\varphi : \{a\}^* \to \Sigma^*$
 by $\varphi(a^i) = w^i$, which is injective as $|w| > 0$ by assumption.
 Then, the unary language $\varphi^{-1}(L) = \{ a^i \mid w^i \in L\}$ is regular, as inverse
 homomorphisms preserve regularity.
 Hence, we can write it as a union of languages recognizable by automata
 with a single final state, which, by Lemma~\ref{lem::unary_single_final},
 have the form $\{ a^i \}$ for some $i \ge 0$
 or $\{ a^{i + jp} \mid j \ge 0 \}$ for some $i \ge 0, p > 0$.
 As the application of functions preserves union, and $L = \varphi(\varphi^{-1}(L))$ here,
 the language 
 $L$ is the union of the images of these languages.
 We have $\varphi(\{a^i \}) = \{ w^i \}$, and this singleton language
 is obviously recognizable by a polycyclic automaton,
 and we have $\varphi(\{ a^{i + jp} \mid j \ge 0 \}) = \{ w^{i+pj} \mid j \ge 0 \}$,
 and this language is also recognizable by an automaton that
 has an initial tail labelled by $w^i$ and a cycle labelled by $w^p$.
 So, as the polycyclic languages are closed under union~\cite[Proposition 6]{DBLP:conf/ictcs/Hoffmann20},
 we have shown that the language $L$ 
 is recognizable by some polycyclic automaton.     \emph{[End, Proof of the Claim]}
 \end{quote}
 Finally, as the languages recognizable by polycyclic automata
 are closed under concatenation and union~\cite[Proposition 5 and Proposition 6]{DBLP:conf/ictcs/Hoffmann20},
 by Theorem~\ref{thm:bounded_regular_form} every bounded regular language is recognizable by a polycyclic automaton.
\end{proof}

In~\cite[Theorem 2]{DBLP:conf/ictcs/Hoffmann20} it was shown that for polycyclic
constraint languages, the constrained problem is always in $\NP$.
So, we can deduce the next result.

\begin{theorem}
\label{thm:sparse_in_NP}
 If $L \subseteq \Sigma^*$ is sparse and regular, then $L\textsc{-Constr-Sync} \in \NP$.
\end{theorem}

We will need the following closure property stated in~\cite[Theorem 3.8]{DBLP:reference/hfl/Yu97}
of the sparse regular languages.

\begin{proposition}
 The class of sparse regular languages is closed under homomorphisms.
\end{proposition}

\begin{toappendix}
Note that sparse languages in general are not closed
under homomorphic mappings~\cite{Pin2020}.
As it is easy to see that the bounded languages are closed
under homomorphic mappings, this also implies that, in general,
the bounded languages do not equal the sparse languages.
\end{toappendix}

Note that the connection of the polycyclic languages to the sparse or bounded languages
was not noted in~\cite{DBLP:conf/ictcs/Hoffmann20}. However, a condition
characterizing the sparse regular languages
in terms of forbidden patterns was given in~\cite{Pin2020}, and
it was remarked that ``a minimal deterministic automaton recognises a sparse language if and only if it
does not contain two cycles reachable from one another''.
This is quite close to our characterization.

\section{Letter-Bounded Constraint Languages}
\label{sec:strictly_bounded_case}






 Fix a constraint automaton $\mathcal B = (\Sigma, P, \mu, p_0, F)$.
 Let $a_1, \ldots, a_k \in \Sigma$ be a sequence of (not necessarily distinct)
 letters.
 In this section, we assume $L(\mathcal B) \subseteq a_1^* \cdots a_k^*$.
 A language which fulfills the above condition
 is called \emph{letter-bounded}.
 Note that the language $ab^*a$ given in the introduction as an example
 is letter-bounded. In fact, it is the language with the smallest
 recognizing automaton yielding an \NP-complete constrained problem~\cite{DBLP:conf/mfcs/FernauGHHVW19}.
 
 A language such that the $a_i$ are pairwise distinct, i.e., $a_i \ne a_j$
 for $i \ne j$, is called \emph{strictly bounded}.
 The class of strictly bounded languages has been extensively studied~\cite{DBLP:journals/mst/BlattnerC77,DBLP:journals/dam/DassowP99,Ginsburg66,GinsburgSpanier64,GinsburgSpanier66,HerrmannKMW17},
 where in~\cite{Ginsburg66,GinsburgSpanier64,GinsburgSpanier66} no name was introduced for them
 and in~\cite{HerrmannKMW17} they were called strongly bounded.
 The class of letter-bounded languages properly contains the strictly bounded languages.
 
 \begin{toappendix} 
 Note that the work~\cite{HerrmannKMW17}
 seems to deviate from the standard terminology, for example by calling bounded languages
 as introduced here word-bounded and refers to letter-bounded simply as bounded languages.
 \end{toappendix}

 \begin{remark}
 \label{rem:motivation_strictly_bounded}
Let $\Sigma = \{b_1, \ldots, b_r\}$
be an alphabet of size $r$.
Then, 
the mappings
\[
\Phi(L) = L \cap b_1^* \cdots b_r^* \mbox{ and }
\perm(L) = \{ w \in \Sigma^* \mid \exists u \in L\  \forall a \in \Sigma : |u|_a = |w|_a \}
\]
for $L \subseteq \Sigma^*$ are mutually inverse and inclusion preserving
between the languages in $b_1^* \cdots b_r^*$ and the commutative languages
in $\Sigma^*$, where a language $L \subseteq \Sigma^*$ is commutative 
if $\perm(L) = L$.
Furthermore, for strictly bounded languages of the form $B_1 \cdots B_r \subseteq b_1^* \cdots b_r^*$
with $B_j \subseteq \{b_j\}^*$, $j \in \{1,\ldots, r\}$, we have
$
 \perm(B_1 \cdots B_r) = B_1 \shuffle \cdots \shuffle B_r,
$
where
$U \shuffle V = \{ u_1 v_1 \cdots u_n v_n \mid u_i, v_i \in \Sigma^*, u_1 \cdots u_n \in U, v_1 \cdots v_n \in  V \}$ for $U, V \subseteq \Sigma^*$.
Hence, $\perm(L)$ is regularity-preserving
for strictly bounded languages.
More specifically, the above correspondence between
the two language classes is regularity-preserving in both directions.
For commutative constraint languages, a classification
of the complexity landscape has been achieved~\cite{DBLP:conf/cocoon/Hoffmann20}.
By the close relationship between commutative and certain strictly
bounded languages, it is natural to tackle
this language class next.
However, as shown in~\cite{DBLP:conf/cocoon/Hoffmann20},
for commutative constraint languages, we can realize $\PSPACE$-complete
problems, but, by Theorem~\ref{thm:sparse_in_NP},
for strictly bounded languages, the constrained problem is always in $\NP$.
However, by the above relations, Theorem~\ref{thm:bounded_regular_form} for languages in $b_1^* \cdots b_r^*$ is equivalent to~\cite[Theorem 5]{DBLP:conf/cocoon/Hoffmann20}, a representation result
for commutative regular languages.
\end{remark}

 Our first result says, intuitively,  
 that if in $A_1 \cdots A_k$ with $A_j$ unary and regular,
 if no infinite unary language $A_j$ over $\{a_j\}$ lies 
 between 
 non-empty unary languages
 over a distinct letter\footnote{Hence different from $\{\varepsilon\}$, as $\{\varepsilon\} \subseteq \{a\}^*$
 for $a \in \Sigma$.}  than $a_j$, 
 then $(A_1 \cdots A_k)$\textsc{-Constr-Sync} is in~$\PTIME$.
 
 
\begin{propositionrep}
\label{prop:stricly_bounded_P}
 Let $A_j \subseteq \{a_j\}^*$ be unary regular languages
 for $j \in \{1,\ldots, k\}$.
 Set $L = A_1 \cdot\ldots\cdot A_k$.
 If for all $j \in \{1,\ldots, k\}$, $A_j$ infinite implies that $A_i \subseteq \{a_j\}^*$
 for all $i < j$ or $A_i \subseteq \{a_j\}^*$ for all $i > j$ (or both), then $L\textsc{-Constr-Sync} \in \PTIME$. 
\end{propositionrep} 
\begin{proof}
 Let $L = A_1 \cdots A_k$ with $A_j \subseteq \{a_j\}^*$ fulfill the assumption.
 If $A_j$ is infinite and for all $i < j$ we have $A_i \subseteq \{a_j\}^*$,
 then $A_1 \cdots A_j \subseteq \{a_j\}^*$, and similarly if
 for all $i > j$ we have $A_i \subseteq \{a_j\}^*$.
 So, by considering $(A_1 \cdots A_j) A_{j+1} \cdots A_k$
 or $A_1 \cdots A_{j-1} (A_j \cdots A_k)$, with $j$ maximal in the former case and minimal in the latter,
 without loss of generality, we can assume $j = 1$
 or $j = k$, i.e., we only have the cases $A_1$
 is infinite, $A_k$ is infinite or both are infinite or none is infinite,
 and, by maximality or minimality of $j$,
 in all these cases the languages $A_2, \ldots, A_{k-1}$ are all finite.

 Then, by Lemma~\ref{lem:union_single_final_state}, we can write $A_1$ and $A_k$
 as a finite union of unary languages recognizable by automata with a single final state.
 As concatenation distributes over union, if we do this for
 $A_1$ and $A_k$ and rewrite the language using the mentioned distributivity,
 we get a finite union of languages of the form
 \[
  A_1' A_2 \cdots A_{k_1} A_k'
 \]
 where $A_1'$ and $A_k'$ are recognizable by unary automata with a single final state
 and are either finite or infinite. Hence, by Lemma~\ref{lem:union},
 if we can show that the problem is in $\PTIME$ for each such language,
 the result follows. 
 So, without loss of generality, we assume
 from the start that $A_1$ or $A_k$ are recognizable by automata
 with a single final state.

 If all $A_j$, $j \in \{1,\ldots,k\}$ are finite, then $L$ is finite, and $L\textsc{-Constr-Sync}\in \PTIME$
 by Lemma~\ref{lem:finite}.
 We handle the remaining cases separately.
 
 \begin{enumerate}
 \item[(i)] Only $A_1$ is infinite.
  
  By assumption, every $A_j \subseteq \{a_j\}^*$, $j \in \{1,\ldots,k\}$,
  is recognizable by a single state automaton. Hence, by Lemma~\ref{lem::unary_single_final}, we can write, as $A_1$ is infinite,
   $A_1 = a_1^i (a_1^p)^*$ with $i \ge 0$ and $p > 0$.
  Let $\mathcal A = (\Sigma, Q, \delta)$ be an input semi-automaton
  for $L\textsc{-Constr-Sync}$.
  As $\delta(Q, a_1) \subseteq Q$,
  we have, for any $n \ge 0$, $\delta(Q, a_1^{n+1}) \subseteq \delta(Q, a_1^n)$.
  So, as $Q$ is finite and the sequence of subsets
  cannot get arbitarily small, for some $0 \le n < |Q|$
  we have $|\delta(Q, a_1^{n+1})| = |\delta(Q, a_1^n)|$.
  But $|\delta(Q, a_1^{n+1})| = |\delta(Q, a_1^n)|$,
  as $\delta(Q, a_1^{n+1}) \subseteq \delta(Q, a_1^n)$,
  implies $\delta(Q, a_1^{n+1}) = \delta(Q, a_1^n)$.
  Then, the symbol $a_1$
  permutes the set $\delta(Q, a_1^n)$.
  Hence, $\delta(Q, a_1^{n+m}) = \delta(Q, a_1^n)$ for any $m \ge 0$.
  So, combining these observations,
  \begin{equation}\label{eqn:case_one_P}
   \{ \delta(Q, a_1^n) \mid n \ge 0 \} = \{ \delta(Q, a_1^n) \mid n \in \{0,\ldots, |Q|-1\} \}
  \end{equation}
  and $\delta(Q, a_1^{|Q| - 1 + m}) = \delta(Q, a_1^{|Q|-1})$
  for any $m \ge 0$. 
  Now, note that the language $A_2 \cdots A_k$
  is finite. So, to find out if we have any
  $a_1^{i + lp} u$ with $u \in A_2 \cdots A_k$
  that synchronizes the input semi-automaton,
  we only have to test if any of the words
  $
   a_i^{i + lp} u,
  $
  with $u \in A_2 \cdots A_k$
  and $l$ such that $i + lp \le \max\{|Q|-1 + p, i\}$,
  synchronizes $\mathcal A$.
  The number (and the length) of these words is linear bounded in $|Q|$ 
  and each could be checked in polynomial time by 
  feeding it into the input semi-automaton for each state and checking
  if a unique state results.
  Hence the problem is solvable in polynomial time.
   
 \item[(ii)] Only $A_k$ is infinite.
 
  Let $u \in A_1 \cdots A_{k-1}$. By assumption, there are only finitely many such
  words $u$. Set $S = \delta(Q, u)$ and $T = \delta(Q, a_k^{|Q|-1})$.
  As in case (i), $a_k$ permutes the states in $T$
  and as $S \subseteq Q$, we have  $\delta(S, a_k^{|Q|- 1}) \subseteq T$.
  So, as $a_k$ permutes $T$, it acts injective on the
  subset $\delta(S, a_k^{|Q|- 1})$.
  This gives $|\delta(S, a_k^{|Q|- 1 + n})| = |\delta(S, a_k^{|Q|- 1})|$
  for any $n \ge 0$. Together with $|\delta(S, a_k^{n + 1})| \le |\delta(S, a_k^{n})|$,
  we have
  \begin{equation}\label{eqn:case_two_P}
   \exists n \in \mathbb N_0 : |\delta(S, a_k^{n})| = 1 \Leftrightarrow |\delta(S, a_k^{|Q|- 1})| = 1.
  \end{equation}
  Choose any fixed $N \ge |Q| - 1$ with $a_k^N \in A_k$.
  Then, with the above considerations, we only have to test the finite
  number of words
  \[
   u\cdot a_k^{N}, \quad u \in A_1 \cdots A_{k-1}.
  \]
  The length of these words is linear bounded in $|Q|$ and 
  as each test, i.e., feeding the word into the input semi-automaton
  for each state and testing if a unique state results,
  could be performed in polynomial time, the problem is solvable in polynomial time.
  
 \item[(iii)] Both $A_1$ and $A_k$ are infinite.
  
  This is essentially a combination of the arguments of case (i) and (ii).
  Let $\mathcal A = (\Sigma, Q, \delta)$ be an input semi-automaton
  and $\mathcal B = (\Sigma, P, \mu, p_0, F)$
  be a constraint automaton with $L = L(\mathcal B)$.
  First, consider only the language $A_1 \cdots A_{k-1}$.
  Then, as in case (i), see Equation~\eqref{eqn:case_one_P},
  \[
   \{ \delta(Q, a_1^n) \mid a_1^n \in A_1  \}
    = \{ \delta(Q, a_1^n) \mid 0 \le n < |Q| - 1 + |P|\mbox{ and } a_1^n \in A_1 \}.
  \]
  Note that we have written $0 \le n < |Q| - 1 + |P|$
  and not merely $\le |Q| - 1$ as an upper bound.
  The reason is that otherwise, if $a_1^{|Q|-1} \notin A_1$,
  we might miss the set $\delta(Q, a_1^{|Q|-1})$,
  but as $\delta(Q, a_1^{|Q|-1+m}) = \delta(Q, a_1^{|Q|-1})$
  for any $m \ge 0$ and $A_1$ is infinite, $\delta(Q, a_1^{|Q|-1}) \in \{ \delta(Q, a_1^n) \mid a_1^n \in A_1  \}$.
  However, if $a_1^n \in A_1$ for some $n \ge |Q| - 1 + |P|$,
  then, with $s = \mu(p_0, a_1^{|Q| - 1})$,
  by finiteness of $P$, among
  the states $s, \mu(s,a_1), \ldots, \mu(s, a_1^{n - |Q| + 1})$
  we find $0 \le m \le |P| - 1$ and $0 < r \le |P|$ with $m + r \le |P|$
  such that $\mu(s, a_1^{m+r}) = \mu(s, a_1^m)$.
  Then we have found a cycle and we can skip it, i.e.,
  \begin{align*}
   \mu(p_0, a_1^n) & = \mu(s, a_1^{n - |Q| + 1}) 
                     = \mu(\mu(s, a_1^{m+r}), a_1^{n - |Q| + 1 - (m+r)}) \\
                   & = \mu(\mu(s, a_1^m), a_1^{n - |Q| + 1 - (m+r)}) \\
                   & = \mu(s, a_1^{m + n - |Q| + 1 - m - r}) \\
                   & = \mu(p_0, a_1^{n-r}).
  \end{align*}
  But, as then $\mu(s, a_1^{n - r}) = \mu(s, a_1^n) \in F$
  we find $a_1^{n-r} \in A_1$. 
  Repeating this procedure, if $n - r \ge |Q| - 1 + |P|$,
  we ultimately find $|Q| - 1 \le m < |Q| - 1 + |P|$
  such that $a_1^m \in A_1$
  and $\delta(Q, a_1^{|Q| - 1}) = \delta(Q, a_1^m)$.
  Note that the language $A_2 \cdots A_{k-1}$
  is finite.
  Then, as in case (i), 
  we only have to consider the  words,
  whose length and number is linear bounded in $|Q|$,
  \[
   a_1^n \cdot u,\quad  0 \le n < |Q| - 1 + |P|, a_1^n \in A_1, u \in A_2 \cdots A_{k-1}
  \]
  and the corresponding sets
  \[
   S = \delta(Q, a_1^n \cdot u),
  \]
  and these are all possible sets in $\{ \delta(Q, a_1^n u) \mid a_1^n \in A_1, u \in A_2 \cdots A_{k-1} \}$.
  Fix any such subset $S$.
  Then, as in case (ii) and Equation~\eqref{eqn:case_two_P}, 
  choose any $N \ge |Q| - 1$ with $a_k^N \in A_k$ and
  we only have to compute $\delta(S, a_k^N)$
  and test if it is a singleton set.
  So, in total, we only have to test the words
  \[
   a_1^n \cdot u a_k^N, 0 \le n < |Q| - 1 + |P|, a_1^n \in A_1, u \in A_2 \cdots A_{k-1}.
  \]
  Their length and number is linear bounded in $|Q|$
  and computing the reachable state from each state of the input automaton,
  and testing if a unique state results, could be performed in polynomial
  time. Hence, the overall procedure could be performed in polynomial time.
 \end{enumerate}
 So, we have handled every case and the proof is complete.\qed
\end{proof}

Now, we state a sufficient condition for \NP-hardness over binary alphabets. 
This condition, together with Proposition~\ref{prop:hom_lower_bound_complexity},
allows us to handle the general case in Theorem~\ref{thm:dichotomy}.
Its application together with Proposition~\ref{prop:hom_lower_bound_complexity} shows, in some respect, that the language $ab^*a$
is the prototypical language giving \NP-hardness.
We give a proof sketch of Lemma~\ref{lem:np_hardness} at the end of this section.

\begin{toappendix}

In the proof of Lemma~\ref{lem:np_hardness}
we will need the following two lemmata.
For $n > 0$, set
\[
 L_n = (\Sigma^* a \Sigma^* b^{|P|} \Sigma^*)^n.
\]
Recall that $\mathcal B = (\Sigma, P, \mu, p_0, F)$.

\begin{lemma}
\label{lem:number_of_b_blocks}
 Let $\Sigma = \{a,b\}$ and $L(\mathcal B) \subseteq a_1^* \cdots a_k^*$
 with $a_i \in \Sigma $ and $n > 0$.
 Then, the following are equivalent:
 \begin{enumerate} 
 \item $L(\mathcal B) \cap L_n \ne \emptyset$,
 
 \item there exist $u_0, \ldots, u_n \in \Sigma^* a \Sigma^*$
  and $p_1, \ldots, p_n \ge |P|$
  such that \[
  u_0 b^{p_n} u_1 \cdots u_{n-1} b^{p_n} u_n \in L(\mathcal B),
  \]
 \item there exist $u_0, \ldots, u_n \in \Sigma^* a \Sigma^*$
  and $p_1, \ldots, p_n > 0$
  such that 
  \[ 
  u_0 (b^{p_1})^* u_1 \cdots u_{n-1} (b^{p_n})^* u_n \subseteq L(\mathcal B).
  \]
 \end{enumerate}
\end{lemma}
\begin{proof}
 That (1) implies (2) is obvious.
 As $p_i \ge |P|$, when reading these factors they have to induce a loop in $\mathcal B$,
 which implies (3).
 Lastly, if (3) holds true, as
 \[
  u_0 b^{|P|\cdot p_1} u_1 \cdots u_{n-1} b^{|P| \cdot p_n} u_n \in L(\mathcal B)
 \]
 and $u_i \in \Sigma^* a \Sigma^*$,
 we also find $u_0 b^{|P|\cdot p_1} u_1 \cdots u_{n-1} b^{|P| \cdot p_n} u_n \in L_n$
 and (1) follows.\qed
\end{proof}

\begin{lemma}
\label{lem:maximal_n}
 Let $\Sigma = \{a,b\}$ and $L(\mathcal B) \subseteq a_1^* \cdots a_k^*$
 with $a_i \in \Sigma $.
 Then, there exists a maximal $n$
 such that $L(\mathcal B) \cap L_n \ne \emptyset$
 and for this maximal $n$,
 we can assume that $u_i \notin \Sigma^* b^{|P|} \Sigma^*$
 for the $u_i$, $i \in \{0,\ldots,n\}$, as in the previous lemma
 and $n \le |P|$. 
\end{lemma}
\begin{proof}
 Recall $\mathcal B = (\Sigma, P, \mu, p_0, F)$
 Note that $\mathcal B$ must necessarily be polycylic (this is a slightly stronger
 claim than Theorem~\ref{thm:bounded_characterization}, as this theorem
 only asserts existence of some polycyclic automaton)
 after removing all states that are not coaccessible, i.e., states from which no final state is reachable,
 which could obviously be done without altering $L(\mathcal B)$.
 For if $\mathcal B$ is then not polycyclic, then some strongly
 connected component does not consists of a single cycle only
 and we find two distinct words $u, v$ and a state $p \in P$
 such that $\mu(p, u) = \mu(p, v) = p$ (see also the forbidden
 pattern in~\cite[Theorem 4.29]{Pin2020}).
 But then, if we choose $x,y \in \Sigma^*$
 such that $\mu(p_0, x) = p$
 and $\mu(p, y) \in F$, we find $x(u+v)^*y \subseteq L(\mathcal B)$.
 Set $m = \max\{|u|, |v|\}$
 Then, for $i > 0$,
 \[
  \{ w \in x(u+v)^*y : |w| \le |x| + i \cdot m\}
 \]
 contains $x(u+v)^i$, and $|x(u+v)^i| = 2^i$.
 So, $L(\mathcal B) \cap \{ w \in \Sigma^* : |w| \le n \}$
 contains at least $2^{\lfloor n - (|x| - |y|) / m \rfloor}$
 many words, i.e., it not sparse.
 Furthermore, as $L \cap \Sigma^n \in O(n^c)$
 as a function of $n$ if and only if $L \cap \{ w \in \Sigma^* : |w| \le n \} \in O(n^{c'})$
 as a function of $n$ for some $c,c' \ge 0$,
 the claim follows.
 
 So, we can assume $\mathcal B$ is polycyclic and every state is coaccessible.
 Now, note that this implies that every loop in $\mathcal B$ (or strongly connected component
 in this case) must be labelled by a single letter, for if we
 have $\mu(p, u) = p$ with $|u|_a > 0$ and $|u|_b > 0$
 and choose again $x,y$ such that $\mu(p_0, x) = p$
 and $\mu(p, y) \in F$,
 we find $xu^ky \in L(\mathcal B)$, which contradict $L(\mathcal B) \subseteq a_1^* \cdots a_k^*$.
 
 But then, note that if, for example, $aba \in L(\mathcal B)$,
 we must have $|P| \ge 2$, as $\mu(p_0, ab) \notin \{ p_0, \mu(p_0,a) \}$.
 Similarly, if we have a word that switches letters, every time a letter-switch
 occurs the state we end up in $\mathcal B$ must be a new state not visited before,
 for otherwise we would have a loop whose transition are not exclusively
 labelled by a single letter.
 
 So, this implies that 
 if we have a word as written in Lemma~\ref{lem:number_of_b_blocks}
 in $L(\mathcal B)$, then $n \le |P|$
 which implies that we can find a maximal $n$.
 That $u_i \notin \Sigma^* b^{|P|} \Sigma^*$
 is also implied by Lemma~\ref{lem:number_of_b_blocks}
 and the maximality of $n$. \qed
\end{proof}
\end{toappendix}

\begin{lemmarep}
\label{lem:np_hardness}
 Suppose $\Sigma = \{a,b\}$. 
 Let $L(\mathcal B) \subseteq \Sigma^*$ be letter-bounded.
 Then, $L(\mathcal B)$\textsc{-Constr-Sync}
 is $\NP$-hard 
 if $L(\mathcal B) \cap \Sigma^* a  b^{|P|}b^*  a \Sigma^* \ne \emptyset$.
\end{lemmarep}
\begin{toappendix}
\begin{figure}[htb]
     \centering
     \hspace*{-2.5cm}
\includegraphics[width=17cm]{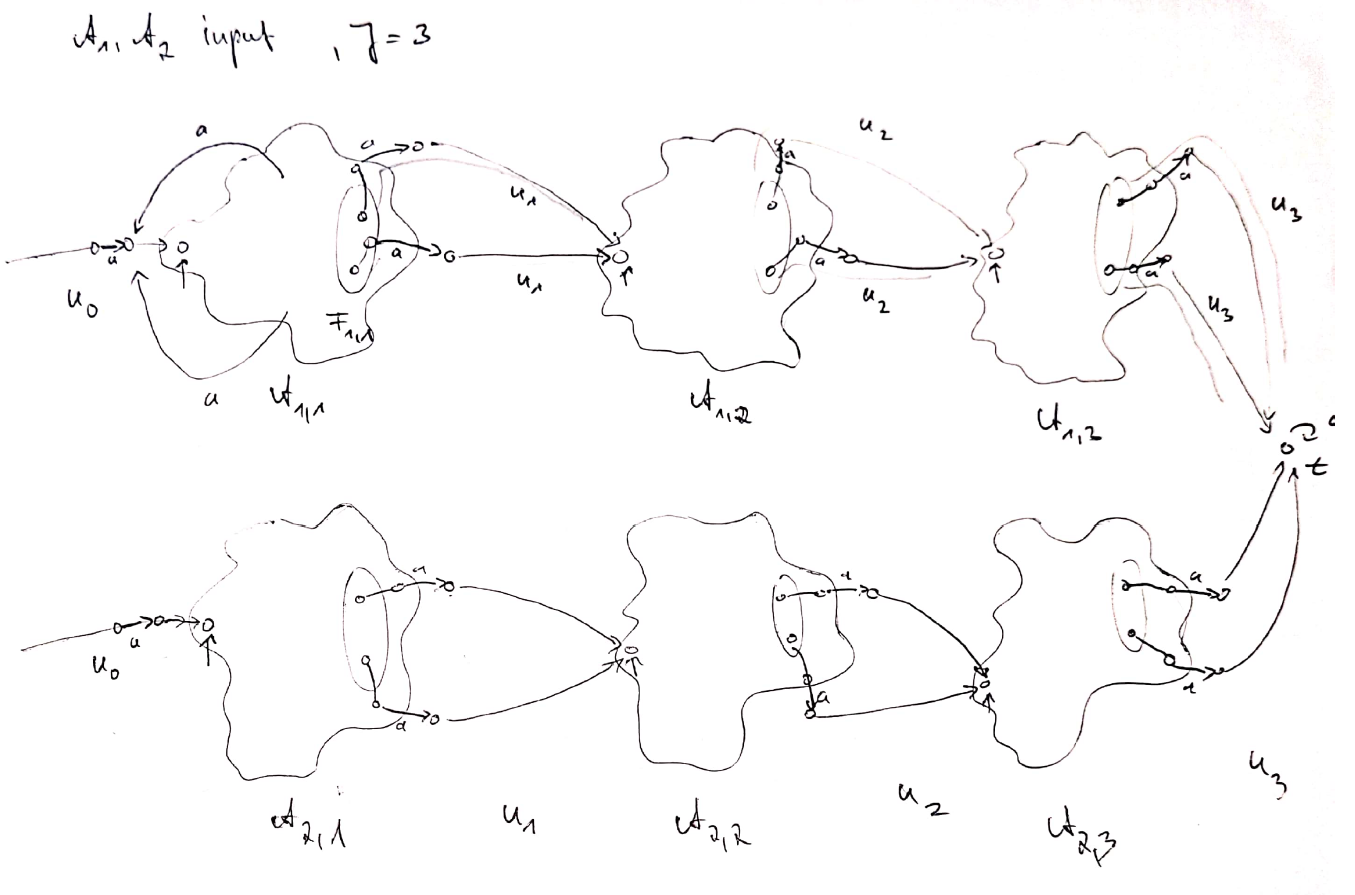}
  \caption{
   The reduction from the proof of Lemma~\ref{lem:np_hardness}
   in the special case $J = 3$ (see the proof for the definition of $J$)
   and two input automata $\mathcal A_1, \mathcal A_2$ over $\{b\}$. The automata
   $\mathcal A_{i,j}$ are inflated, according to Definition~\ref{def:inflate_aut},
   copies of $\mathcal A_i$
   for $i \in \{1,2\}$, $j \in \{1,2,3\}$. 
   The letter $a$ maps
   every state not associated with a path inside each $\mathcal A_{i,j}$ to the last innermost state that 
   is hit by an $a$ along the path leading into this automaton. This is only drawn for $\mathcal A_{1,1}$ but left
   out for the other automata, also, to give a more ``high-level'' drawing, the $b$-transitions
   are not drawn. On the right end is the sink state $t$. The paths stay inside the automata but leave
   as soon as an $a$ is read.}
  \label{fig:reduction}
\end{figure}

\begin{proof}[Proof of Lemma~\ref{lem:np_hardness}]
First, using Lemma~\ref{lem:maximal_n},
choose $J > 0$ maximal such that
\[
 L(\mathcal B) \cap (\Sigma^* a \Sigma^* b^{|P|} \Sigma^*)^J \ne \emptyset.
\]
As stated in the lemma, we have $J \le |P|$ (which implies the constrution to follow could
be carried out in polynomial time).
Then, by Lemma~\ref{lem:number_of_b_blocks}, there
exist $u_0, \ldots, u_J \in \Sigma^* a \Sigma^*$
and $p_1, \ldots, p_J > 0$ ($J > 0$) such that \todo{Anderer Bezeichner als $J$?}
\[
 u_0 (b^{p_1})^* u_1 \cdots u_{J-1} (b^{p_J})^* u_J \subseteq L(\mathcal B).
\]
 Let $N$ be $|P|$ times the least common multiple of the numbers $p_1, \ldots, p_J$.
 We give a reduction from the {\sc DFA-Intersection} for unary 
 input automata, which is \NP-complete in this case~\cite{stockmeyer1973word,fernau2017problems}.
 Let $\mathcal A_i = (\{b\}, Q_i, \delta_i, q_i, F_i)$
 for $i \in \{1,\ldots,k\}$ be unary input automata, and we want to know
 if they all accept a common word. The problem remains
 \NP-complete if we assume for no input automaton, a start state is also a final state.
 This is easily seen but could also be shown similar to~\cite[Proposition 1]{DBLP:conf/ictcs/Hoffmann20}.
 Also, we can assume $F_i \ne \emptyset$ for all $i \in \{1,\ldots,k\}$.

 We are going to construct a semi-automaton $\mathcal C = (\{a,b\}, Q, \delta)$. 

 Write $u_i = u_{i,1} \cdots u_{i,|u_i|}$ with $u_j \in \Sigma$.
 For each $i \in \{0,\ldots,J\}$, we construct a path labelled with $u_i$.
 Formally, let $P_i = \{ q_{i, 0}, \ldots, q_{i,|u_i|} \} \subseteq Q$
 be new states and set
 \[
  \delta(q_{i,j-1}, u_j) = q_{i,j}.
 \]
 Then, for each $\mathcal A_i$
 we construct $J$ (disjoint) 
 copies of $\mathcal A_i$ and inflate
 them according to Definition~\ref{def:inflate_aut} by $N$.
 Call the results $\mathcal A_{i, 1}, \mathcal A_{i,2}, \ldots, \mathcal A_{i,J}$
 with $\mathcal A_{i,j} = (\{b\}, Q_{i,j}, \delta_{i,j}, s_{i,j}, F_{i,j})$.
 Note these are unary automata over the letter $b$.
 Also, let $t \in Q$ be a new state, which will be a (global) sink state in $\mathcal C$, i.e.,
 we set $\delta(t, a) = \delta(t, b) = t$.
 Next, we describe how we interconnect these automata with the paths and with $t$.
 See also Figure~\ref{fig:reduction} for a sketch of the reduction in the special case $J = 3$
 and two input automata.

 \begin{enumerate}
 \item Let $j \in \{1,\ldots, J\}$. For each final state $q \in F_{i,j}$
  let $P_{i,q}$ be a disjoint copy of the path $P_i$ constructed above,
  except for one final state $q$ were we simply retain the path $P_i$, but also name it by $P_{i,q}$.
  By identifying states, we mean states that we have previously constructed are now merged
  to a single state in $Q$. We have to pay attention that this procedure does not introduces
  any non-determinism.
  We identify the state $q_{i,0}$ with $q$
  and continue to identify the states $q_{i,j}$ and $q' \in Q_{i,j}$
  if $q_{i,j-1}$ and $q'' \in Q_{i,j}$ were identified
  and $u_{i,j} = b$ and $q' = \delta_{i,j}(q'', b)$. As $u_i \in \Sigma^* a \Sigma^*$, this process has to come to a halt
  before we have identified $< J$ states.
  Note that the first state such that $q_{i,j-1}$ and $q''\in Q$
  were identified but not $q_{i,j}$ and $\delta_{i,j}(q'', b)$, i.e., were $u_{i,j} = a$,
  we have added an $a$-transition to $q_{i,j}$
  from $q'' = q_{i,j-1}$ in $\mathcal A_{i,j}$, i.e., this is the first
  $a$-transition we have added to $\mathcal A_{i,j}$ and it branches out of $\mathcal A_{i,j}$.
  
  Then, if $j \le J - 1$,
  identify the state $q_{i,J}$ with the start state $s_{i,j+1}$ of $\mathcal A_{i,j+1}$, i.e.,
  the path $P_{i,q}$ ends at this state.
  And if $j = J$, we identify the state $q_{i,J}$ with $t$.
  
 \item For the path $P_0$ identify its end state $q_{i,|u_0|}$
  with the start state $s_{i,1}$ of $\mathcal A_{i,1}$.
     
 \item Up to now, we still hav emissing transitions. In all the paths created, 
  every missing $b$-transition, i.e., were we have a state with an $a$-transition
  leading out but no $b$-transition, we add a self-loop labelled with $b$ to that state.
  For each path $P$ (including the copies constructed in the first step)
  let $p \in P$ be that state closest to the end state, but that does not equal
  the end state (by the identifications above, some end state might already have an $a$-transition
  that goes out of some automaton $\mathcal A_{i,j}$) and has an outgoing $a$-transition.
  Such a state exists as the $u_i \in \Sigma^* a \Sigma^*$.
  \todo{Hier uU statt auf den Zustand immer auf den letzten Zustand davor mit einer $a$-Transition mappen?}
  Then, for every state in $P$ that does not have an $a$-transition
  we add an $a$-transition going to $p$.
  Consider $\mathcal A_{i,j}$ and let $P$ some path (the specific choice does not matter)
  ending at the start state of $\mathcal A_{i,j}$.
  For each state $q \in Q_{i,j}$ that does not has an outgoing $a$-transition up to now,
  add an $a$-transition going to the state $p \in P$ described above in that path.
  This ensures later that, by reading an $a$, we end up in a well-defined situation.
 \end{enumerate}
 Then, put all the states created so far, i.e., those of the $\mathcal A_{i,j}$
 and those of the paths constructed, into $Q$ (note for each $i \in \{1,\ldots,k\}$
 we have constructed paths and automata, intuitively we have copied each $\mathcal A_i$, inflated
 the copies and interconnected them with the paths given by the $u_i$)
 and let $\delta$ be the transition as defined above or as given by $\mathcal A_{i,j}$
 on the state of these automata.

 We need the following property of $\mathcal C$. Suppose $i \in \{1,\ldots,k\}$, 
 $j \in \{1,\ldots,|J|\}$ and $w \in \{a,b\}^*$.
 
  \medskip 
 
\noindent\underline{Claim:}
  Let $q \in Q_{i,j} \setminus F_{i,j}$ with $\delta(q, w) = t$.
  Then, there exist 
  \[ 
  u_1, u_2 \cdots, u_{|J|-i+1} \subseteq \{b\}^*
  \]
  and $u,v\in \{a,b\}^*$
  such that $|u_i| \ge N$ and $|u_i|$ is divisible by $N$
  for all $i \in \{1,\ldots,|J|-i+1\}$
  and $v_1, \ldots, v_{|J|-i+1} \in \{a,b\}^*a\{a,b\}^*$
  so that 
  \[
   w = vu_1 v_1 u_2 v_2 \cdots u_{|J|-i+1} v_{|J|-i+1} u
  \]
  and $v \notin \Sigma^* b^N \Sigma^*$.
 \begin{quote} 
     \emph{Proof of the Claim.} First, the state $q \in Q_{i,j}$ has to be mapped to a final
     state, which could only be done by a word containing at least $N$
     times the letter $b$,
     as in the inflated construction we can only go from non-auxiliary states
     to non-auxiliary states by reading at least that number of letters.
     However, before that we might read some word $v \in \{a,b\}^*$ that moves states around, does
     not has a consecutive sequence of more than $N$ $b$'s and hence, every $a$
     goes back to the start state. But at some point, this has to come to an end and we have
     to read a sequence of more than $N$ consecutive $b$'s.
     Additionally, by the construction of the inflation, the word
     that moves from a non-auxiliary state to another non-auxiliary state
     must have a number of $b$'s that is divisible by $N$.
     Also, observe that such a word must consists entirely of $b$, because
     for non-final states in $\mathcal A_{i,j}$ every $a$ maps back to the start state.
     Then, by construction (recall $u_i \in \Sigma^* a \Sigma^*$ for the labels
     of the paths constructed above) to move between the automata $\mathcal A_{i,j}$
     inside of $\mathcal C$
     we have to traverse a path
     where, on some part, we can only move forward by reading the letter $a$.
     After this, when we are at the start state of $\mathcal A_{i,j+1}$,
     as by assumption the start state is not final, we again have to read at least $N$
     times the letter $b$ and so on, until we have reached a final 
     state in $\mathcal A_{i,|J|}$.
     Then, we have to read at least one $a$ to map the final state to $t$, from which
     on, as $t$ is a sink state, we can read any word. 
     \emph{[End, Proof of the Claim]}
 \end{quote}

 The automaton $\mathcal C$ has a synchronizing word in $L$
 if and only if all the $\mathcal A_i$, $i \in \{1,\ldots,k\}$,
 accept a common word.
 
 \begin{enumerate}
 \item Assume we have a word $b^n$ accepted by all $\mathcal A_i$ for $i \in \{1,\ldots,k\}$. 
 Then, for
 \[
  w = u_0 b^{N\cdot n} u_1 \cdots u_{J-1} b^{N\cdot n} u_J
 \] 
 we have $w \in L$ and $w$ synchronizes $\mathcal A$.
 Note that, after reading $u_{j-1}$,
 the automaton $\mathcal A_{i,j}$
 is either in its start state, or the final $a$ in $u_{j-1}$
 has mapped some state in $\mathcal A_{i,j}$ to a state outside of $Q_{i,j}$.
 So, when reading $b^{N\cdot n}$, 
 as $\mathcal A_{i,j}$ equals  the inflation of $\mathcal A_i$ by $N$,
 we end up in a final state $F_{i,j}$.
 Then, we read $u_j$ to map those final states to the start state
 of the next automaton $\mathcal A_{i,j+1}$ or to $t$ if $j = J$.
 Note that all states in-between are either mapped
 to a start state of some $\mathcal A_{i,j}$, moved inside of some
 $\mathcal A_{i,j}$, or, when an $a$ is read and they are not mapped
 back to a state that ultimately ends in a start state of some $\mathcal A_{i,j}$
 are moved toward the state $t$.
 As we always read enough $a$ to always make a step towards the sink state $t$
 the result follows.\todo{genauer}

 \item  Assume $\mathcal A$ has a synchronizing word $w \in L$.
  Then, as $t$ is a sink state, the word $w$ must map every state to $t$.
  Consider the start state of some $\mathcal A_{i,1} = (\{b\}, Q_{i,1}, \delta_{i,1}, q_{i,1}, F_{i,1})$.
  By the above claim,  
  there exist $u_1, u_2 \cdots, u_{J} \subseteq \{b\}^*$
  such that $|u_i| \ge N$ and $|u_i|$ is divisible by $N$
  for all $i \in \{1,\ldots,|J|\}$
  and $v_1, \ldots, v_{J} \in \{a,b\}^*a\{a,b\}^*$ and $v, u \in \{a,b\}^*$
  so that 
  \[ 
    w = vu_1 v_1 u_2 v_2 \cdots u_{J} v_{J} u.
  \]
  By the above claim, Lemma~\ref{lem:maximal_n} and the maximal choice of $J$,
  we have 
  \[ 
  \{ v, v_1, \ldots, v_J \} \cap \Sigma^* b^{|P|} \Sigma^* = \emptyset,
  \] 
  i.e.,
  these words does not contains a sequence of more than $|P|$, and so in particular not more than $N$,
  consecutive $b$'s.

  Now, let $b^n$ be a maximal non-empty factor whose length $n$ is divisible by $N$ of $vu_1 v_1$
  and using only the letter $b$.
  Note that, by construction of $\mathcal A_{i,1}$,
  if we write $vu_1 v_1 = x b^n y$,
  we have $\delta_{i,1}(q_{i,1}, x) = q_{i,1}$.
  Then, we claim that $b^{n / N}$ is accepted
  by every automaton $\mathcal A_i$. 
  Fix an index $i \in \{1,\ldots,k\}$.
  By the construction of the inflation, this is equivalent
  with the condition that $b^n$ drives every automaton
  $\mathcal A_{i,j}$ for $j \in \{1,\ldots,|J|\}$
  from the start state to some final state.
  Suppose this is not the case. As the automata $\mathcal A_{i,j}$
  are isomorphic, i.e., they are copies of each other, we can assume this is not the case for $\mathcal A_{i,1}$, i.e., 
  we have $\delta_{i,1}(q_{i,1}, b^n) \notin F_{i,1}$.
  Then, consider the following suffix of $w$ (recall $xb^n y = v u_1 v_1$, and $y$ has to start with an $a$) 
  \[
   y u_2 v_2 \cdots u_{|J|} v_{|J|} u.
  \] 
  Note that if we have in $u$ a consecutive sequence of $b$'s
  of length more than $N$, the rest of $u$ also must consist of $b$'s only, i.e.,
  we cannot read an $a$ anymore.
  For suppose this is not the case and $u \in \Sigma^* b^N \Sigma^* a \Sigma^*$.
  We have $\delta(q_{i,1}, xb^n) = \delta(q_{i,1}, b^n) \in Q_{i,1} \setminus F_{i,1}$.
  By assumption, $\delta(q_{i,1}, w) = t$,
  and so we must have $\delta(q_{i,1}, y u_2 v_2 \cdots u_{|J|} v_{|J|} u) = t$.
  Applying the above claim again,
  yields that we can factorize $y u_2 v_2 \cdots u_{|J|} v_{|J|} u$
  such that we have at least $|J|$ blocks of consecutive 
  $b$'s broken up by at least one occurrence of the letter $a$
  between each such block.
  However, then 
  then
  \[
   w = x b^n y u_2 v_2 \cdots u_{|J|} v_{|J|} u,
  \]
  as $y$ starts with an $a$, we would get a factorization
  of $w$ with $|J| + 1$ blocks of consecutive $b$'s separated by words
  with at least one $a$, which is not possible by the maximal choice
  of $J$ and Lemma~\ref{lem:number_of_b_blocks}.
  
 \end{enumerate}
 So, this shows that this is a valid reduction.\qed
\end{proof}
\end{toappendix}

So, finally, we can state our main theorem of this section.
Recall that by Theorem~\ref{thm:sparse_in_NP},
and as the class of bounded regular languages equals
the class of sparse regular languages~\cite{DBLP:journals/eik/LatteuxT84}, for bounded regular constraint
languages, the constrained problem is, in our case, in \NP.

\begin{theoremrep}[Dichotomy Theorem]
\label{thm:dichotomy}
 %
 %
 Let $a_1, \ldots, a_k \in \Sigma$ be a sequence of letters
 and $L \subseteq a_1^* \cdots a_k^*$ be regular.
 The problem $L\textsc{-Constr-Sync}$
 is \NP-complete if
 \[
  L \cap \left(\bigcup_{\substack{1 \le j_1 < j_2 < j_3 \le k \\ a_{j_2} \notin \{a_{j_1}, a_{j_3}\} }} L_{j_1,j_2,j_3} \right) \ne \emptyset
 \]
 with $L_{j_1, j_2, j_3} = \Sigma^* a_{j_1} \Sigma^* a_{j_2}^{|P|} \Sigma^* a_{j_3} \Sigma^*$
 for $1 \le j_1 < j_2 < j_3 \le k$ and solvable in polynomial time otherwise.
\end{theoremrep}
\begin{proof}
 Set $L = L(\mathcal B)$.
 Let $j_1, j_2, j_3 \in \{1,\ldots,k\}$
 be such that $a_{j_2}\notin\{a_{j_1},a_{j_3}\}$, $j_1 < j_2 < j_3$
 and
 \[
  L(\mathcal B) \cap L_{j_1, j_2, j_3} \ne \emptyset.
 \]
 Then, there exists a word $u_1 a_{j_1} u_2 a_{j_2}^{|P|} u_3 a_{j_3} u_4 \in L(\mathcal B)$
 with $u_1, u_2, u_3, u_4 \in \Sigma^*$.
 By the pigeonhole principle, when reading the factor $b^{|P|}$,
 at least one state has to be traversed twice 
 and we find $p > 0$ such that $u_1 a u_2 b^{|P| + i\cdot p} u_3 a u_4$
 for any $i \ge 0$.

 Define a homomorphism $\varphi : \Sigma^* \to \{a,b\}^*$
 by $\varphi(a_{j_1}) = \varphi(a_{j_3}) = a$,
 $\varphi(a_{j_2}) = b$
 and, for the remaining letters, $\varphi(a) = \varepsilon$,
 if $a \in \Sigma \setminus \{a_{j_1}, a_{j_2}, a_{j_3}\}$.
 Then, $\varphi(L) \subseteq \varphi(a_1)^* \cdots \varphi(a_k)^*$
 is letter-bounded. 
 Set $\Gamma = \{a,b\}$ and let $\mathcal B' = (\Gamma, P', \mu', p_0', F')$
 be a recognizing PDFA for $\varphi(L)$.
 We have
 \[
  \varphi(u_1) a \varphi(u_2) b^{|P| + i\cdot p} \varphi(u_3) a \varphi(u_4) \in \varphi(L) 
 \]
 for any $i \ge 0$. So,
 $
 \varphi(L) \cap \Gamma^* a^+ \Gamma^* b^{|P'|}b^* \Gamma^* a^+ \Gamma^* \ne \emptyset.
 $
 By Lemma~\ref{lem:np_hardness}, $\varphi(L)$\textsc{-Constr-Sync}
 is \NP-hard and so, by Proposition~\ref{prop:hom_lower_bound_complexity},
 also $L\textsc{-Constr-Sync}$ is \NP-hard, and so, with Theorem~\ref{thm:sparse_in_NP},
 \NP-complete.

 Now, suppose
  $
 L(\mathcal B) \cap \left(\bigcup_{\substack{1 \le j_1 < j_2 < j_3 \le k \\ a_{j_2} \notin \{a_{j_1}, a_{j_3}\} }} L_{j_1,j_2,j_3} \right) = \emptyset.
 $
  By Theorem~\ref{thm:bounded_regular_form}, we can
 write 
 $L(\mathcal B) = \bigcup_{i=1}^n A_1^{(i)} \cdots A_k^{(i)}$ 
 with unary regular languages $A_j^{(i)} \subseteq \{a_j\}^*$
 for $j \in \{1,\ldots,k\}$.
 Then, 
 \[ 
 ( A_1^{(i)} \cdots A_k^{(i)} ) \cap \left(\bigcup_{\substack{1 \le j_1 < j_2 < j_3 \le k \\ a_{j_2} \notin \{a_{j_1}, a_{j_3}\} }} L_{j_1,j_2,j_3} \right) = \emptyset
 \]
 for any $i \in \{1, \ldots, n\}$.
 However, this implies that for any $i \in \{1,\ldots,n\}$, if there exists $j \in \{1,\ldots, k\}$
 such that $A_j^{(i)}$ is infinite, 
 then for all $j' < j$, or for all $j' > j$, 
 we have $A_{j'} \subseteq \{a_j\}^*$ (recall that if $A_{j'} = \{\varepsilon\}$, then
 this is also fulfilled).
 Hence, by Proposition~\ref{prop:stricly_bounded_P},
 we have $(A_1^{(i)} \cdots A_k^{(i)})\textsc{-Constr-Sync} \in \PTIME$
 and then, by Lemma~\ref{lem:union},
 $L(\mathcal B)\textsc{-Constr-Sync} \in \PTIME$.\qed
\end{proof}

As the languages $L_{j_1, j_2, j_3}$ are regular, we
can devise a polynomial-time algorithm which checks the condition
mentioned in Theorem~\ref{thm:dichotomy}. 
 
\begin{corollary} 
 Given a PDFA $\mathcal B$ and a sequence of letters $a_1, \ldots, a_k$
 as input such that $L(\mathcal B) \subseteq a_1^* \cdots a_k^*$,
 the complexity of $L(\mathcal B)$\textsc{-Constr-Sync}
 is decidable in polynomial-time.
\end{corollary}
\begin{proof}
 An automaton for each $L_{j_1, j_2, j_3}$
 has size linear in~$|P|$. So, by the product automaton construction~\cite{HopUll79}, non-emptiness of
 $L(\mathcal B)$ with each $L_{j_1, j_2, j_3}$
 could be checked in time $O(|P|^2)$.
 Doing this for every $L_{j_1, j_2, j_3}$
 gives a polynomial-time algorithm
 to check non-emptiness of the language written
 in Theorem~\ref{thm:dichotomy}.~\qed
\end{proof}

\begin{example}
 For the following constraint languages CSP is \NP-complete: $ab^*a$,
 $aa(aaa)^*bbb^*d \cup a^*b \cup d^*$, $bbcc^*d^* \cup a$.
 
 For the following constraint languages CSP is in \PTIME: $a^5bd \cup cd^4$,
 $a^5bd \cup cd^*$, $aa^*bbbbcd^* \cup bbbdd^*d$.
\end{example}

\begin{proof}[Proof Sketch for Lemma~\ref{lem:np_hardness}]
 We construct a reduction from an instance
 of $\textsc{DisjointSetTransporter}$\footnote{Note that the problem $\textsc{DisjointSetTransporter}$ is over a unary alphabet, but for $L\textsc{-Constr-Sync}$
 we have $|\Sigma| > 1$. Indeed, we need the additional letters in $\Sigma$.}
 for unary input automata.
 
 To demonstrate the basic idea, we only do the proof
 in the case $L \subseteq a^* b^* a^*$.
 By assumption we can deduce $a^{r_1} b^{r_2} a^{r_3} \in L(\mathcal B)$
 with $p_2 \ge |P|$ and $r_1, r_3 \ge 1$.
 By the pigeonhole principle, in $\mathcal B$, 
 when reading the factor $b^{r_2}$, at least one state has to be traversed twice.
 Hence, we find $0 < p_2 \le |P|$ such that $a^{r_1} b^{r_2 + i\cdot p_2} a^{r_3}
 \subseteq L(\mathcal B)$ for each $i \ge 0$.

Let $\mathcal A = (\{c\}, Q, \delta)$ and $(\mathcal A, S, T)$
be an instance of \textsc{DisjointSetTransporter}.
We can assume $S$ and $T$ are non-empty, as for $S = \emptyset$
it is solvable, and if $T = \emptyset$ we have no solution.
Construct $\mathcal A' = (\Sigma, Q', \delta')$
by setting
$
 Q' = S_{r_2} \cup \ldots \cup S_{1} \cup Q \cup Q_1 \cup \ldots \cup Q_{p_2-1} \cup \{ t \},
$
where $t$ is a new state, $S_i = \{ s_i \mid s \in S \}$ for $i \in \{1,\ldots, r_2 \}$
are pairwise disjoint copies of $S$
and $Q_i = \{ q^i \mid q \in Q \}$ are\footnote{Observe
that by the indices a correspondence between the sets
is implied. The index
in $Q_i$ at the top to distinguish, for $s \in S$ and $i \in \{1,\ldots,\min\{r_2, p_2-1\}\}$, between
 $s_i \in S_i$ and $s^i \in Q_i$. Hence, for each $s \in S$ and $i \in \{1,\ldots, r_2\}$,
 the states $s$ and $s_i$ correspond to each other, and for $q \in Q$
 and $i \in \{1,\ldots, p_2-1\}$ the states $q$ and $q^i$.} 
 also pairwise disjoint 
copies of $Q$. Note that also $S_i \cap Q_j = \emptyset$
for $i \in \{1,\ldots, r_2 \}$ and $j \in \{1,\ldots, p_2-1\}$.
Set $S_0 = S$ 
as a shorthand.
Choose any $\hat s \in S_{r_2}$, then, for $q \in Q$ and $x \in \Sigma$, the transition function is given by
\[
 \delta'(q, x) = \left\{
 \begin{array}{ll}
  s_{i-1} & \mbox{if } x = b \mbox{ and } q = s_i \in S_i \mbox{ for some } i \in \{1,\ldots, r_2\}; \\ 
  \hat s & \mbox{if } x = a \mbox{ and } q \in (Q \cup Q_1 \cup \ldots \cup Q_{p_2-1}) \setminus S; \\
  s_{r_2} & \mbox{if } x= a \mbox{ and } q = s_i \in S_i \mbox{ for some } i \in \{0,\ldots,r_2\}; \\
  t       & \mbox{if } x = a \mbox{ and } q \in T; \\
  q^{p_2-1} & \mbox{if } x = b \mbox{ and } q \in Q; \\
  q^{i-1} & \mbox{if } x = b \mbox{ and } q = q^i \in Q_i \mbox{ for some } i \in \{2,\ldots,p_2-1\}; \\
  \delta(q, c) & \mbox{if } x = b \mbox{ and } q = q^1 \in Q_1; \\
  q       & \mbox{otherwise}.
 \end{array}
 \right.
\]

\newcommand{\automatacloudother}[2][.44]{%
	\begin{scope}[#2]
		\node [rectangle,draw,thick,text width=8.1cm,minimum height=7.6cm,
		text centered,rounded corners, fill=white, name = re] {};
\end{scope}}

\newcommand{\innerstateloop}{
\begin{scope}
  \node[state] (s1) at (0,0) {}; \node (s1label) at (0.5,-.1) {$\in Q$};
  \node[state] (s11) at (-0.5,0.6) {};  \node (s11label) at (0.32,0.6) {$\in Q_{p_2-1}$};
  \node (s12) at (-0.25,1.2) {};
  \node (s13) at ( 0.25,1.2) {};
  \path[->] (s1)  edge [bend left] node [left] {$b$} (s11);
   \path[->] (s11) edge [bend left] node [left] {$b$} (s12);
  \draw[dashed] (-0.25,1.2) -- (0.25,1.2);
\end{scope}
}

\begin{figure}[htb]
     \centering
    \scalebox{.65}{    
 \begin{tikzpicture}
 \tikzset{every state/.style={minimum size=1pt},>=stealth'}
 \node (cloud) at (0,0) {\tikz \automatacloudother{fill=gray!0,thick};};
 
  \node (reset1) at (0,3) {};
  \node (reset2) at (-10,2.5) {};
  \path[->] (reset1) edge [bend right] node [above] {$a$} (reset2);
  
  \node (reset3) at (0,-3) {};
  \node (reset4) at (-10,-2.5) {};
  \path[->] (reset3) edge [bend left] node [below] {$a$} (reset4);
  
  \node (reset5) at (-5.4,2.8) {};
  \node (reset6) at (-10,2.5) {};
  \path[->] (reset5) edge [bend right] node [above,pos=.3] {$a$} (reset6);
  
  \node (reset7) at (-5.4,-2.8) {};
  \node (reset8) at (-10,-2.5) {};
  \path[->] (reset7) edge [bend left] node [below,pos=.3] {$a$} (reset8);

  \draw[rounded corners] (-4.1,-3) rectangle (-1.7, 3) {};
  \draw[rounded corners] (-6.3,-3) rectangle (-4.8, 3) {};
  \draw[rounded corners] (-10.5,-3) rectangle (-9, 3) {};
    
  \draw[rounded corners] (1.7,-3) rectangle (4.1, 3) {};
  
  \node[state] (t) at (7,0) {$t$};
  
  \node at (3,3.4) {{\LARGE $T$}};
  \node at (-3,3.4) {{\LARGE $S$}};
  \node at (-10,3.5) {{\LARGE $S_{r_2}$}};
  \node at (-5.5,3.5) {{\LARGE $S_{1}$}};
  \node at (0.1,4.1) {{\LARGE Original $\mathcal A$ (altered)}};
  
  \node (s1) at (-.5,2) {\tikz \innerstateloop;};
  \node (s2) at (.5,-1.5) {\tikz \innerstateloop;};
  
  \node (sT1) at (3,1.9) {\tikz \innerstateloop;}; \node (sT1copy) at (3,1.5) {};
  \node (sT2) at (2.8,0.4) {\tikz \innerstateloop;}; \node (sT2copy) at (2.8,0) {};
  \node (sT3) at (3.1,-1.5) {\tikz \innerstateloop;}; \node (sT3copy) at (3.15,-2) {};
  
  \node (sS1) at (-3,1.7) {\tikz \innerstateloop;}; \node (sS1copy) at (-2.95,1.15) {};
  \node (sS2) at (-2.7,0) {\tikz \innerstateloop;}; \node (sS2copy) at (-2.65,-.55) {};
  \node (sS3) at (-3,-1.8) {\tikz \innerstateloop;};\node (sS3copy) at (-3,-2.3) {};
   
  \node[state] (sS11) at (-5.5,1.7) {};
  \node[state] (sS12) at (-5.2,0) {};
  \node[state] (sS13) at (-5.7,-2) {};
  
  \node[state] (sSr1) at (-10,1.7) {};
  \node[state] (sSr2) at (-9.7,0) {};
  \node[state] (sSr3) at (-10.1,-2) {};
  
  \path[->] (t) edge [loop right] node {$\Sigma$} (t);
  
  \path[->] (sT1copy) edge [bend left=35] node [above] {$a$} (t)
            (sT2copy) edge [bend left=10] node [above] {$a$} (t)
            (sT3copy) edge node [above] {$a$} (t);

  \node (sSr1a) at (-8.5,1.7) {};
  \node (sSr2a) at (-8.2,0) {};
  \node (sSr3a) at (-8.7,-2) {};
  
  \node (sSr1b) at (-7.2,1.7) {};
  \node (sSr2b) at (-6.9,0) {};
  \node (sSr3b) at (-7.4,-2) {};
  
  \path[->] (sSr1) edge node [above,pos=.3] {$b$} (sSr1a);
  \path[->] (sSr2) edge node [above,pos=.27] {$b$} (sSr2a);
  \path[->] (sSr3) edge node [above] {$b$} (sSr3a);
 
  \path[->] (sSr1b) edge node [above,pos=.35] {$b$} (sS11);
  \path[->] (sSr2b) edge node [above] {$b$} (sS12);
  \path[->] (sSr3b) edge node [above] {$b$} (sS13);
  
  \path[->] (sS11) edge [bend right=10] node [above,pos=.38] {$b$} (sS1copy);
  \path[->] (sS12) edge [bend right=10] node [above,pos=.2] {$b$} (sS2copy);
  \path[->] (sS13) edge [bend right=10] node [above,pos=.4] {$b$} (sS3copy);
  
  \draw[dashed] (-8.5,1.7) -- (-7.2,1.7);
  \draw[dashed] (-8.2,0) -- (-6.9,0);
  \draw[dashed] (-8.7,-2) -- (-7.4,-2);
 \end{tikzpicture}}
  \caption{
   The reduction from the proof sketch sketch of Lemma~\ref{lem:np_hardness}.
   The letter $a$ transfers everything surjectively onto $S_{r_2}$,
   indicated by four large arrows at the top and bottom and labelled 
   by $a$.
   The auxiliary states $Q_1, \ldots, Q_{p_2-1}$, which are meant
   to interpret a sequence $b^{p_2}$ like a single symbol in the original
   automaton, are also only indicated inside of $\mathcal A$, but not fully written out.}
  \label{fig:reduction_np_hard}
\end{figure}

Please see Figure~\ref{fig:reduction_np_hard} for a sketch
of the reduction.
For the constructed automaton $\mathcal A'$, the following could be shown:
$\exists m \ge 0 : \delta(S, c^m) \subseteq T$
if and only if $\mathcal A'$ has a synchronizing word in $ab^{r_2}(b^{p_2})^*a$
if and only if $\mathcal A'$ has a synchronizing word in $ab^*a$
if and only if $\mathcal A'$ has a synchronizing word in $a^*b^*a^*$.

\begin{toappendix}

Next, we supply the proof of the claim made in the proof sketch of Lemma~\ref{lem:np_hardness}
from the main text.

\medskip

\noindent\underline{Claim:} 
 For the constructed automaton $\mathcal A'$ from
 the proof sketch of Lemma~\ref{lem:np_hardness} in the main text, we have:
\begin{align*}
    \exists m \ge 0 : \delta(S, c^m) \subseteq T 
                               & \Leftrightarrow \mathcal A'\mbox{ has a synchronizing word in $ab^{r_2}(b^{p_2})^*a$.} \\
                               & \Leftrightarrow \mathcal A'\mbox{ has a synchronizing word in $ab^*a$.} \\
                               & \Leftrightarrow \mathcal A'\mbox{ has a synchronizing word in $a^*b^*a^*$}
\end{align*} 
 \emph{Proof of the Claim.}
 First, suppose $\delta(S, c^m) \subseteq T$.
 By construction of $\mathcal A'$,
 for any $q, q' \in Q$,
 \begin{equation}\label{eqn:transition_Astar}
  \delta(q, c) = q'  \mbox{ in $\mathcal A$}  \Leftrightarrow \delta'(q, b^{p_2}) = q'  \mbox{ in $\mathcal A'$} .
 \end{equation}
 Also, $\delta'(Q'\setminus\{t\},a) = S_{r_2}$
 and $\delta'(S_{r_2}, b^{r_2}) = S$.
 Combining these facts, we find
 \[
  \delta'(Q', ab^{r_2}b^{p_2m}) \subseteq T \cup \{t\}. 
 \]
 A final application of $a$ then maps
 all states in $T$ to the single sink 
 state~$t$.
 
 Clearly, as $ab^{r_2}(b^{p_2})^*a \subseteq a b^* a$
 and $a b^* a \subseteq a^* b^* a^*$, the next two implications are shown.
 Finally, to complete the argument, let $u = a^{p} b^q a^r$ be a synchronizing word, $p,q,r \ge 0$.
 Then, as $t$ is a sink state, $\delta'(Q', u) = \{t\}$.
 The only way to enter $t$ from the outside is to read $a$ at least once, and 
 as $t$ is a sink state, we have $\delta'(Q', a^p b^q a^r) = \{t\}$.
 Also, as for $q \notin T$, we have $\delta'(q, a) \notin T$,
 we must have $\delta'(Q', a^p b^q) \subseteq T \cup \{t\}$,
  or more specifically, $\delta'(Q' \setminus \{t\}, a^p b^q) \subseteq T$.
  We distinguish two cases for~$p$.
 
 \begin{enumerate}
 \item If $p = 0$, then, in particular, $\delta'(S, b^q)  \subseteq T$.
     By construction of $\mathcal A'$, for any $q \in Q$,
 \[
  \delta'(q, b^n) \in Q
 \]
 if and only if $n \equiv 0\pmod{p_2}$.
 So, $q = p_2 m$ for some $m \ge 0$.
 Hence, by Equation~\eqref{eqn:transition_Astar} above from the first case, in $\mathcal A$,
 we find $\delta(S, c^m) \subseteq T$.
 
 \item  If $p > 0$, then $\delta'(Q' \setminus\{t\}, a^p) = S_{r_2}$.
 The only way to leave any state in $S_{r_2}$
 is to read $b$, which transfers $S_{r_2}$ to $S_{r_2-1}$.
 Reasoning similarly, we find that we have to read in $b$
 at least $r_2$ many times, which finally maps $S_{r_2}$
 onto $S_0 = S$. So, $q \ge r_2$. By construction of $\mathcal A'$, for any $q \in Q$,
 \[
  \delta'(q, b^n) \in Q
 \]
 if and only if $n \equiv 0\pmod{p_2}$.
 So, as $\delta'(S, b^{q - r_2}) \subseteq T$, $q - r_2 = p_2 m$ for some $m \ge 0$.
 Hence, by Equation~\eqref{eqn:transition_Astar} above, in $\mathcal A$,
 we find $\delta(S, c^m) \subseteq T$.
 \end{enumerate}
This ends the proof of the claim. \emph{[End, proof of the Claim.]}
\end{toappendix}

 Now, suppose $\delta(s, c^m) \subseteq T$ for some $m \ge 0$.
 By the above, $\mathcal A'$ 
 has a synchronizing word $u$ in $ab^{r_2}(b^{p_2})^*a$.
 Then, $a^{r_1 - 1}u a^{r_3-1} \in L(\mathcal B)$ also synchronizes~$\mathcal A'$.

 Conversely, suppose we have a synchronizing word $w \in L$
 for $\mathcal A'$.
 As $L \subseteq a^* b^* a^*$
 by the above equivalences,
 $\delta(S, c^m) \subseteq T$
 for some $m \ge 0$. \qed
\end{proof}

\section{Constraints from Strongly Self-Synchronizing Codes}
\label{sec:strongly_self_sync}
%

Here, we introduce strongly self-synchronizing codes and investigate $L$\textsc{-Constr-Sync}
for bounded constraint languages $L \subseteq w_1^* \cdots w_k^*$
where $\{ w_1, \ldots, w_k \}$ is such a code. 

Let $C \subseteq \Sigma^+$ be non-empty.
Then, $C$ is called a \emph{self-synchronizing code}~\cite{zbMATH03943051,DBLP:books/daglib/0025093,Hsieh1989SomeAP},
if $C^2 \cap \Sigma^+ C \Sigma^+ = \emptyset$. If, additionally, $C \subseteq \Sigma^n$
for some $n > 0$, then it is called\footnote{In~\cite{Hsieh1989SomeAP}
 this distinction is not made and self-synchronizing codes are also called comma-free codes.}
a \emph{comma-free code}~\cite{golomb_gordon_welch_1958}.
Every self-synchronizing code is an infix code, i.e., no proper factor of a word from $C$ is in $C$~\cite{Hsieh1989SomeAP}.
A \emph{strongly self-synchronizing code}
is a self-synchronizing code $C \subseteq \Sigma^+$ \todo{nicht bloss $\cap C\Sigma^+$ möglich? für beweis ausreichend?}
such that, additionally, $(\pref(C) \setminus C)C \cap \Sigma^*C \Sigma^+ = \emptyset$.

To give some intuition for the strongly self-synchronizing
codes, we also present an alternative characterization, a few examples and a way to construct such codes.

\begin{propositionrep}
A non-empty $C \subseteq \Sigma^+$
is a strongly self-synchronizing code
if and only if, for
all $u \in \pref(C)$ and $v \in C$, 
if we write $uv = x_1 \cdots x_n$
with $x_i \in \Sigma$ for $i \in \{1,\ldots, n\}$,
then, for all $j \in \{1,\ldots,n\}$ and $k \ge 1$ where $j + k - 1 \le n$,
we have that $x_j x_{j + 1} \cdots x_{j+k-1} \in C$
implies $j = |u| + 1$ and $k = |v|$
or $j = 1$ and $k = |u|$.
Intuitively, in $uv$ only the last $|v|$ symbols form a factor in $C$
and possibly the first $|u|$ symbols.
\end{propositionrep}
\begin{proof}
 Let $C \subseteq \Sigma^+$ be a strongly self-synchronizing code.
 Suppose $u \in \pref(C)$ and $v \in C$.
 If $u \notin C$, then we must have $uv \notin \Sigma^*C\Sigma^+$,
 so that, if $uv = x_1 \cdots x_n$ as in the statement,
 we have $x_j \cdots x_{j+k-1}$ if and only if $j = |u| + 1$
 and $k = |u| + |v|$.
 If $u \in C$, then, as $uv \notin \Sigma^+ C \Sigma^+$,
 we find that we have only the possibilities
 $j = 1$ and $k = |u|$ or $j = |u| + 1$ and $k = |u| + |v|$.
 
 Conversely, suppose $C \subseteq \Sigma^+$ is non-empty
 and fulfills the condition mentioned in the statement.
 If $u, v \in C$ and $uv \in \Sigma^+ C \Sigma^+$,
 then we can write $uv = x_1 \cdots x_n$ with $x_i \in \Sigma$ for $i \in \{1,\ldots,n\}$
 and find $2 \le i \le j \le n - 1$
 such that $x_i \cdots x_j \in C$, which contradicts
 the condition in the statement.
 Similarly, if $u \in \pref(C) \setminus C$
 and $v \in C$ with $uv \in \Sigma^*C\Sigma^+$,
 then we can write $uv = x_1 \cdots x_n$ with $x_i \in \Sigma$ for $i \in \{1,\ldots,n\}$
 and find $1 \le i \le j \le n - 1$
 such that $x_i \cdots x_j \in C$, which would contradict
 the condition too. So, we must
 have $C^2 \cap \Sigma^+ C \Sigma^+ = \emptyset$
 and $(\pref(C) \setminus C) C \cap \Sigma^* C \Sigma^+ = \emptyset$.\qed 
\end{proof}


 When passing from letters to words by applying a homomorphism, in the reductions,
 we have to introduce additional states. The definition of the strongly synchronizing
 codes was motivated by the demand that these states also have to be synchronized, which turns out to be difficult in general.

\begin{example}\label{ex:strongly_self_sync}
 The code $\{aacc,bbc,bac\}$
 is strongly self-synchronizing.
 The code $\{ aab, bccc, abc \}$ is self-synchronizing, but
 not strongly self-synchronizing as, for example, $(a)(abc)$ 
 contains $aab$ or $(aa)(bccc)$ contains $abc$. 
\end{example}

\begin{toappendix}
 To give a proof of the claim made in Example~\ref{ex:strongly_self_sync}.
\begin{proposition}
 The code $\{aacc,bbc,bac\}$ is strongly self-synchronizing.
\end{proposition}
\begin{proof}
 By checking all cases to combine prefixes with code words:
 \[ 
 \begin{array}{llll}
     \mbox{Non-empty prefixes of $aacc$:} & (a)aacc & (a)bbc & (a)bac \\ 
      & (aa)aacc  & (aa)bbc  & (aa)bac \\ 
      & (aac)aacc & (aac)bbc & (aac)bac \\ 
      & (aacc)aacc & (aacc)bbc & (aacc)bac \\ 
      \\
     \mbox{Non-empty prefixes of $bbc$:} & (b)aacc   & (b)bbc   & (b)bac \\    
     & (bb)aacc  & (bb)bbc  & (bb)bac \\ 
     & (bbc)aacc & (bbc)bbc & (bbc)bac \\ 
     \\
     \mbox{Non-empty prefixes of $bac$:} & (b)aacc   & (b)bbc   & (b)bac \\    
     & (ba)aacc  & (ba)bbc  & (ba)bac \\ 
     & (bac)aacc & (bac)bbc & (bac)bac. 
 \end{array}
 \]
 So, we see that the defining conditions are satisfied. Note that $C \cap \Sigma^*C\Sigma^+ = \emptyset$
 is always satisfied for self-synchronizing codes, as they are infix codes. \qed
\end{proof}
\end{toappendix}

\begin{remark}[Construction] 
\label{rem:code_construction}
 Take any non-empty finite language $X \subseteq \Sigma^n$, $n > 0$,
 and a symbol $c \in \Sigma$ such that $\{c\}\Sigma^* \cap X = \emptyset$.
 Let $k=\max\{\,\ell\geq0\mid \exists u,v\in\Sigma^*:uc^\ell v\in X\,\}$.
 Then, $Y = c^{k+1}X$
 is a strongly self-synchronizing code.
\end{remark}

\begin{example}\label{ex:strongly_self_sync_construction}
Let $\Sigma = \{a,b,c\}$
and $C = \{ ab,ba, aa\}$.
Then, $\{ cab, cba, caa \}$ or $\{ bbab, bbaa \}$
are strongly self-synchronizing codes by Remark~\ref{rem:code_construction}.
\end{example}

Our next result, which holds in general, states conditions on a homomorphism
such that we not only have a reduction from the problem
for the homomorphic image to our original problem, as stated in Proposition~\ref{prop:hom_lower_bound_complexity},
but also a reduction in the other direction.

\begin{theoremrep}
\label{thm:constr_sync_hom_strongly_self_sync}
 Let $\varphi : \Sigma^* \to \Gamma^*$
 be a homomorphism such that $\varphi(\Sigma)$
 is a strongly self-synchronizing code and $|\varphi(\Sigma)| = |\Sigma|$.
 Then, for each regular $L \subseteq \Sigma^*$ we have
 $
  L\textsc{-Constr-Sync} \equiv_m^{\log} \varphi(L)\textsc{-Constr-Sync}.
 $
\end{theoremrep}
\begin{proof}
 By Proposition~\ref{prop:hom_lower_bound_complexity}, we have 
 $\varphi(L)\textsc{-Constr-Sync} \le_m^{\log} L\textsc{-Constr-Sync}$.

 Next, we give a reduction 
 from $L\textsc{-Constr-Sync}$ 
 to $\varphi(L)\textsc{-Constr-Sync}$.
 %
 %
 %
 
 %
 %
 %
 %
 Write $\Sigma = \{a_1, \ldots, a_n\}$ with $n = |\Sigma|$
 and $u_i = \varphi(a_i)$ for $i \in \{1,\ldots,n\}$.
 Let $\mathcal A = (\Sigma, Q, \delta)$
 be an input semi-automaton for $L\textsc{-Constr-Sync}$.
 
  We construct a semi-automaton $\mathcal A' = (\Gamma, Q', \delta')$.
  The state set will be
  \[
   Q' = \{ q_x \mid q \in Q, x \in \pref(\{ u_1,\ldots,u_n \}) \setminus \{ u_1,\ldots,u_n \} \}. 
  \]
  By identifying $q_{\varepsilon}$ with the state $q \in Q$,
  we can assume $Q \subseteq Q'$.
  Then, for $q_x \in Q'$ and $y \in \Sigma$, 
  let $z$ be the longest suffix of $xy$
  such that $z \in \pref(\{ u_1,\ldots,u_n\})$ and set\footnote{Note
  the implicit correspondence between the states $q$
  and $q_z$ for $z \in \pref(\varphi(\Sigma)) \setminus \varphi(\Sigma)$.}
  \begin{equation}\label{eqn:def_delta_bar}
   \delta'(q_x, y) = \left\{
   \begin{array}{ll} 
    q_{z}          & \mbox{if } z \in \pref(\{ u_1,\ldots,u_n\}) \setminus \{ u_1, \ldots, u_n \}; \\ 
    \delta(q, a_i)  & \mbox{if } \exists i \in \{ 1, \ldots, n \} : z = u_i.
   \end{array}
   \right.
  \end{equation}
  As $|\varphi(\Sigma)| = |\Sigma|$ and $\{ u_1, \ldots, u_n \}$ is a prefix code\footnote{A code
  is a prefix code, if no code word is the proper prefix of another code word.}, the transition function 
  is well-defined.
  By construction, for any $u \in \Sigma^*$ and $q \in Q$,
  we have  
  \begin{equation}\label{eqn:comma_free_reduction}
      \delta(q, u) = \delta'(q, \varphi(u)).
  \end{equation}

  Let $x \in \pref(\varphi(\Sigma)) \setminus \varphi(\Sigma)$
  and $u_i \in \varphi(\Sigma)$, $i \in \{1,\ldots,n\}$.
  Then, as $\varphi(\Sigma)$
  is a strongly self-synchronizing code, the word $xu_i$ does not contain
  a word from $\varphi(\Sigma)$ , except the suffix $u_i$,
  as a factor. Next, we will argue that, for the unique $a_i \in \Sigma$
  with $\varphi(a_i) = u_i$, the following equations 
  holds true:
  \begin{equation}\label{eqn:strongly_self_sync_transition}
   \delta'(q_x, u_i) = \delta'(q, u_i) = \delta(q, a_i).
  \end{equation}
  For if $v \in \pref(\{ u_i \}) \cap \Sigma$, then the longest suffix of $xv$
  in $\pref(\varphi(\Sigma))$ must be~$v$.
  First, it is a suffix
  from this set.
  Second, if there exists longer one, say $w$,
  then write $ww' \in \varphi(\Sigma)$ for some $w' \in \Sigma^*$.
  In that case, with $xv = x'w$ ($|x'| < |x|$), we have $x'ww' \in xu_i\Sigma^*$
  or $xu_i \in x'ww'\Sigma^+$.
  In the first case, $ww'$ contains the proper factor $u_i \in \varphi(\Sigma)$,
  which is not possible as $\varphi(\Sigma)$ is, in particular, an infix code.
  In the second case, $\{ x'u_i \} \cap \Sigma^* \varphi(\Sigma) \Sigma^+ \ne \emptyset$,
  which is excluded by the property of $\varphi(\Sigma)$ being strongly self-synchronizing.
  So, by the defining equation of $\delta'$, Equation~\eqref{eqn:def_delta_bar},
  if $v \notin \varphi(\Sigma)$, we have
  \[
   \delta'(q_x, v) = q_v,
  \]
  and if $v \in \varphi(\Sigma)$, then $v = u_i$, as $\varphi(\Sigma)$ is a prefix code,
  and
  \[
   \delta'(q_x, v) = \delta'(q_x, u_i) = \delta'(q, u_i) = \delta(q, a_i)
  \]
  with the unique $a_i \in \Sigma$ as above.
  So, in the latter case Equation~\eqref{eqn:strongly_self_sync_transition}
  was established. In the former case,
  if $u_i = vv'v''$, then $\delta'(q_v, v') = q_{vv'}$
  which is easily seen as we always read in a word giving a prefix from $\varphi(\Sigma)$, hence
  this word itself is the longest suffix from $\varphi(\Sigma)$.
  So, after reading the entire word $u_i$, by Equation~\eqref{eqn:def_delta_bar},
  Equation~\eqref{eqn:strongly_self_sync_transition}
  is implied.

  Lastly, we show that this gives a valid reduction.
  
  \begin{myclaiminproof}
   The automaton $\mathcal A = (\Sigma, Q, \delta)$
   has a synchronizing word in $L$
   if and only if $\mathcal A' = (\Gamma, Q', \delta')$
   has a synchronizing word in $\varphi(L)$.
  \end{myclaiminproof}
  \begin{myclaimproof}
   First, suppose there exists $u \in L$ such that $|\delta(Q, u)| = 1$.
    If $|Q| = 1$ every word is synchronizing and the statement is obviously true.
    So, we can assume $|Q| > 1$, which implies $|u| > 0$.
    Write $u = av$ with $a \in \Sigma$.
    By Equation~\eqref{eqn:strongly_self_sync_transition}, then, for any $x \in \pref(\varphi(\Sigma))\setminus\varphi(\Sigma)$,
    \[
     \delta'(q_x, \varphi(a)) = \delta(q, a).
    \]
    Hence, $\delta'(Q', \varphi(a)) = \delta(Q, a)$.
    As $\delta'(Q', \varphi(a)) \subseteq Q$, by Equation~\eqref{eqn:strongly_self_sync_transition},
    or its formulation for the special case of states in $Q$, Equation~\eqref{eqn:comma_free_reduction},
    we find
    \[
     \delta'(\delta(Q, a), \varphi(v))) = \delta(\delta(Q, a), v) = \delta(Q, u).
    \]
    The last set is, by assumption, a singleton set. Hence, the word $\varphi(u)$
    synchronizes~$\mathcal A'$. 
   
    \medskip 
    
    Now, suppose there exists $u \in \varphi(L)$ such that $|\delta'(Q', u)| = 1$.
     Let $v \in \Sigma^*$ be such that $\varphi(v) = u$.
     By Equation~\eqref{eqn:strongly_self_sync_transition} (or Equation~\eqref{eqn:comma_free_reduction}),
     we have
     \[
      \delta(Q, v) = \delta'(Q, \varphi(v)).
     \]
     By assumption, the set on the right side is a singleton set. 
     Hence, $v$ synchronizes $\mathcal A$.
  \end{myclaimproof}
  So, we find $L\textsc{-Constr-Sync} \le_m^{\log} \varphi(L)\textsc{-Constr-Sync}$
  and the proof is done.
\end{proof}

Finally, we apply Theorem~\ref{thm:constr_sync_hom_strongly_self_sync} to bounded languages
such that $\{ w_1, \ldots, w_k\}$ forms a strongly self-synchronizing code.

\begin{theoremrep}
 Let $L \subseteq w_1^* \cdots w_k^*$
 be regular such that $\{ w_1, \ldots, w_k \}$
 is a strongly self-synchronizing code.
 Then, $L\textsc{-Constr-Sync}$
 is either $\NP$-complete or in $\PTIME$. 
\end{theoremrep}
\begin{proof}
 Let $\Gamma = \{ a_1, \ldots, a_n \}$
 be a new alphabet and let $\varphi : \Gamma^* \to \Sigma^*$
 be the homomorphism given by
 $\varphi(a_i) = w_i$ for $i \in \{1,\ldots, n\}$.
 Let $U = \varphi^{-1}(L)$. 
 As every word in $L$
 is a concatentation of words from $\{ w_1, \ldots, w_n \}$,
 we have $L \subseteq \varphi(\Gamma^*)$.
 So, we find $\varphi(U) = L$.

 By Theorem~\ref{thm:constr_sync_hom_strongly_self_sync},
 the languages $U$ and $L$
 have the same computational complexity.
 Also, as is easy to check, we have $U \subseteq a_1^* \cdots a_n^*$
 and $U$ is regular.
 So, by Theorem~\ref{thm:dichotomy}
 the constrained synchronization problem for $L$
 is either $\NP$-complete or in $\PTIME$.
\end{proof}

\begin{example}
 (1) $((aacc)(bbc)^*(bac))$\textsc{-Constr-Sync} is \NP-complete. \\
 (2) $((bbc)(aacc)(bac)^* \cup (bbc)^*)$\textsc{-Constr-Sync} is in \PTIME.
\end{example}

\section{Conclusion and Discussion}

We have looked at the constrained synchronization problem (Problem~\ref{def:problem_L-constr_Sync}) -- CSP for short -- for letter-bounded regular constraint languages and bounded languages induced by strongly self-synchronizing codes, thereby
continuing the investigation started in~\cite{DBLP:conf/mfcs/FernauGHHVW19}.
The complexity landscape in these cases is completely understood.
Only the complexity classes $\PTIME$ and $\NP$-complete arise.
In~\cite{DBLP:conf/ictcs/Hoffmann20} the question was raised if we can find sparse constraint languages
that give constrained problems complete for some candidate $\NP$-intermediate complexity class. At least for the
language classes investigated here
this is not the case. 
For general sparse regular languages, it is still open if a corresponding
dichotomy theorem holds, or candidate $\NP$-intermediate problems arise. By the results obtained so far and the methods
of proofs, we conjecture that in fact a dichotomy
result holds true.

%


Let us relate our results to the previous work~\cite{DBLP:conf/ictcs/Hoffmann20}, where
partial results for \NP-hardness and containment in \PTIME\  were given.
Namely, by setting $\factor(L) = \{ v \in \Sigma^* \mid \exists u,w \in \Sigma^* : uvw \in L \}$
and $\mathcal B_{p,E} = (\Sigma, P, \mu, q, E)$
for $\mathcal B = (\Sigma, P, \mu, p_0, F)$ with $E \subseteq P$ and $q \in P$,
the following was stated.

\begin{proposition}[\cite{DBLP:conf/ictcs/Hoffmann20}]
\label{prop:NPc}
 Suppose we find $u, v \in \Sigma^*$ 
 such that we can write
$
 L = u v^* U
$
 for some non-empty language $U \subseteq \Sigma^*$
 with 
 $
  u \notin \factor(v^*), 
  v \notin \factor(U) \mbox{ and } 
  \pref(v^*) \cap U = \emptyset.
 $
 Then $L\textsc{-Constr-Sync}$ is $\NP$-hard.
\end{proposition}

\begin{proposition}[\cite{DBLP:conf/ictcs/Hoffmann20}]
\label{prop:NP_in_P}
  Let $\mathcal{B} = (\Sigma, P, \mu, p_0, F)$ be a polycyclic PDFA.
  If for every reachable $p \in P$ with $L(\mathcal B_{p, \{p\}}) \ne \{\varepsilon\}$ 
  we have $L(\mathcal B_{p_0, \{p\}}) \subseteq \suff(L(\mathcal B_{p, \{p\}}))$,
  then the problem $L(\mathcal B)\textsc{-Constr-Sync}$ is solvable
  in polynomial time.
\end{proposition}

Note that Proposition~\ref{prop:NPc} implies that $ab^*a$
gives an \NP-complete CSP. However, in the letter-bounded
case there exist constraint languages giving \NP-complete problems
for which this is not implied by Proposition~\ref{prop:NPc},
for example: $ab^*ba$, $ab^*ab$, $aa^*abb^*a$ or $ba^*b \cup a$.
Also, Proposition~\ref{prop:NP_in_P}
is weaker than our Proposition~\ref{prop:stricly_bounded_P}
in the case of letter-bounded constraints.
For example, it does not apply to $ab^*b$, every PDFA for this languages
has a loop exclusively labelled by the letter~$b$
and reachable after reading the letter $a$ from the start state, and
so words along this loop cannot have a word starting with $a$ as a suffix.

For general bounded languages, let us note the following implication of Propositions~\ref{prop:hom_lower_bound_complexity}
and~\ref{prop:stricly_bounded_P}.

\todo{Sind thin languages eigentlich die in $w^*$?}

\begin{propositionrep}
 Let $u,v \in \Sigma^*$. If $L \subseteq u^* v^*$ is regular, then $L$\textsc{-Constr-Sync} is solvable
 in polynomial time.
\end{propositionrep}
\begin{proof}
 Let $\Gamma = \{a,b\}$
 and $\varphi : \Gamma^* \to \Sigma^*$
 be the homomorphism given by $\varphi(a) = u$
 and $\varphi(b) = v$.
 Define $N = \{ (i,j) \mid u^i v^j \in L \}$
 and set $L' = \{ a^i b^j \mid (i,j) \in N \} \subseteq a^* b^*$.
 Then, $\varphi(L') = L$
 and by Proposition~\ref{prop:stricly_bounded_P}
 we have $L'\textsc{-Constr-Sync} \in \PTIME$.
 So, with Proposition~\ref{prop:hom_lower_bound_complexity}
 also $L\textsc{-Constr-Sync} \in \PTIME$.~\qed
\end{proof}

Next, in Proposition~\ref{prop:np_complete_case}, we give an example
of a bounded regular language yielding an $\NP$-complete synchronization problem,
 but for which this is
 not directly implied by the results we have so far.

 
 \begin{propositionrep}\label{prop:np_complete_case}
  The problem $((ab)(ba)^*(ab))$\textsc{-Constr-Sync} is $\NP$-complete.
 \end{propositionrep}
 \begin{proof}
 We give a reduction from $\textsc{DisjointSetTransporter}$
 for unary input semi-automata, which is $\NP$-complete
 by Theorem~\ref{prop:set_transporter_np_complete}.
 Let $\mathcal A = (\{c\}, Q, \delta)$
 with $S, T \subseteq Q$ being disjoint.
 Construct the automaton $\mathcal A' = (\{a,b\}, Q', \delta')$
 with
 \[
  Q' = Q \cup \{ q_a \mid q \in Q \} \cup \{ q_b \mid q \in Q \} \cup \{ t \}. 
 \]
 Fix some $\hat s \in S$.
 Then, for $q \in Q$, set
\begin{align*}
    \delta'(t,   x) & = t \mbox{ for } x \in \Sigma \mbox{ and }
    \delta'(q,   x) = q_x \mbox{ for } x \in \Sigma; \\
    \delta'(q_b, b) & = t   \mbox{ for } q_b \in Q'  \mbox{ and } 
    \delta'(q_b, a) = \delta(q, c); \\
    \delta'(q_a, a) & = q_a \mbox{ for } q_a \in Q'; \\
    \delta'(q_a, b) & = \left\{ 
    \begin{array}{ll}
     \hat s & \mbox{if } q \in Q \setminus (T \cup S); \\
     q      & \mbox{if } q \in S; \\
     t      & \mbox{if } q \in T.
    \end{array}
    \right.
\end{align*}
 Then, there exists $n \ge 0$
 with $\delta(S, c^n) \subseteq T$
 if and only if $\mathcal A'$ has a synchronizing word in $L$.
 
 First, suppose there exists $n \ge 0$
 such that $\delta(S, c^n) \subseteq T$.
 By construction, $S \subseteq \delta'(Q', ab) \subseteq S \cup \{t\} \cup Q_b$,
 or more precisely $\delta'(Q', ab) = S \cup \{t\} \cup \{ q_b \mid q \in \delta(Q, c) \}$.
 Note that, as $S$ and $T$ are disjoint,
 we must have $n > 0$.
 As, for any $q \in Q$, $\delta'(q, ba) = \delta(q, c)$
 and $\delta(q_b, b) = t$,
 we find $\delta'(\delta'(Q', ab), (ba)^n) \subseteq T \cup \{t\}$,
 where we needed $n > 0$ to map those states in $\{ q_b \mid q \in \delta(Q, c) \}$
 to $T$.
 Finally, $\delta(T \cup \{t\}, ab) = \{t\}$
 and so $\delta'(Q', ab(ba)^nab) = \{t\}$.

 Conversely, suppose there exists $n \ge 0$
 such that $\delta'(Q', ab(ba)^nab)$
 is a singleton set. So, as $t$ is a sink state,
 $\delta'(Q', ab(ba)^nab) = \{ t \}$.
 By construction, a state in $Q'$ is mapped to $t$
 by $ab$
 if and only if it is contained in $T \cup \{t\}$.
 Hence, $\delta'(Q', ab(ba)^n) \subseteq T \cup \{t\}$.
 As before, $\delta'(Q', ab) = S \cup \{t\} \cup \{ q_b : q \in \delta(Q, c) \}$.
 In particular, we must have $\delta'(S, (ba)^n) \subseteq T \cup \{t\}$.
 As, for any $q \in Q$, $\delta'(q, ba) = \delta(q, c)$,
 this implies that $\delta'(S, (ba)^n) \subseteq T$
 and that for $u = c^n$ we have $\delta(S, c^n) \subseteq T$.

 By Theorem~\ref{thm:sparse_in_NP}, $L\textsc{-Constr-Sync}\in \NP$
 and by the above reduction the problem is $\NP$-complete.
 \end{proof}

 By Proposition~\ref{prop:np_complete_case},
 for the homomorphism $\varphi : \{a,b\}^* \to \{a,b\}^*$
 given by $\varphi(a) = ab$ and $\varphi(b) = ba$
 both problems $ab^*a$ and $\varphi(ab^*a) = ab(ba)^*ab$
 are \NP-complete. So, this is a homomorphisms
 which preserves, in this concrete instance, the computational complexity.
 But its image $\{ab,ba\}$ is not even a self-synchronizing code.\todo{Ich glaube in dem fall schon. die reduktion sollte immer gehen...}
 However, I do not know if this homomorphism always preserves the complexity.
 Similary, I do not know
 if the condition from Theorem~\ref{thm:constr_sync_hom_strongly_self_sync}
 characterizes those homomorphisms which preserve the complexity.

 In the reduction used in Lemma~\ref{lem:np_hardness}
 the resulting automaton has a sink state. However, in general, for questions
 of synchronizability it makes a difference if we have a sink state
 or not, at least with respect to the \v{C}ern\'y conjecture~\cite{Cer64},
 as for automata with a sink state this conjecture holds true,
 even with the better bound\footnote{In~\cite{DBLP:journals/tcs/Rystsov97}
 erroneously the bound $n(n+1)/2$ was reported as being sharp, but the overall argument
 in fact works to yield the sharp bound $n(n-1)/2$.}
 $\frac{n(n-1)}{2}$~\cite{DBLP:journals/tcs/Rystsov97,DBLP:journals/tcs/Volkov09}. However,
 even in~\cite{DBLP:conf/mfcs/FernauGHHVW19}
 certain reductions establishing \PSPACE-completeness
 use only automata with a sink state. Hence, for hardness
 these automata are sufficient at least in certain instances.
 So, it might be interesting to know
 if in terms of computational complexity of the CSP,
 we can, without loss of generality, limit ourselves to input automata
 with a sink state. The methods of proof for the letter-bounded constraints
 show that in this case, we can actually do this, as these input automata
 are sufficient to establish all cases of intractability.

 Lastly, let us mention the following related problem\footnote{This was actually suggested
 by a reviewer of a previous version.} one could come up with.
 Fix a deterministic and complete semi-automaton~$\mathcal A$.
 Then, for input PDFAs~$\mathcal B$, what is the computational complexity to determine
 if $\mathcal A = (\Sigma, Q, \delta)$ has a synchronizing word in $L(\mathcal B)$?
 As the set of synchronizing words 
 $ \{ w \in \Sigma^* : |\delta(Q, w)| = 1 \} = \bigcup_{q \in Q} \bigcap_{q' \in Q} L((\Sigma, Q, \delta, q', \{q\})) $
 is a regular language, we have to test
 for non-emptiness of intersection of this fixed regular language 
 with $L(\mathcal B)$. This could be done in \NL, hence in \PTIME.

\smallskip \noindent {\footnotesize
\textbf{Acknowledgement.} I thank  anonymous reviewers
of a previous version for detailed feedback.
I also sincerely thank the reviewers of the current version (at least one reviewers saw both versions) for careful reading and giving valuable feedback to improve my scientific writing
and pointing to two instances were I overlooked, in retrospect, two simple conclusions.}

\bibliographystyle{splncs04}
\bibliography{ms} 

\begin{thebibliography}{HKMW17}

\bibitem[FK17]{fernau2017problems}
Henning Fernau and Andreas Krebs.
\newblock Problems on finite automata and the exponential time hypothesis.
\newblock {\em Algorithms}, 10(1):24, 2017.

\bibitem[GKRS10]{DBLP:journals/ijfcs/GawrychowskiKRS10}
Pawel Gawrychowski, Dalia Krieger, Narad Rampersad, and Jeffrey~O. Shallit.
\newblock Finding the growth rate of a regular or context-free language in
  polynomial time.
\newblock {\em Int. J. Found. Comput. Sci.}, 21(4):597--618, 2010.

\bibitem[HKMW17]{HerrmannKMW17}
A.~Herrmann, M.~Kutrib, A.~Malcher, and M.~Wendlandt.
\newblock Descriptional complexity of bounded regular languages.
\newblock {\em Journal of Automata, Languages and Combinatorics},
  22(1-3):93--121, 2017.

\bibitem[Hof19]{DBLP:conf/cai/Hoffmann19}
Stefan Hoffmann.
\newblock Commutative regular languages - properties and state complexity.
\newblock In Miroslav Ciric, Manfred Droste, and Jean{-}{\'{E}}ric Pin,
  editors, {\em Algebraic Informatics - 8th International Conference, {CAI}
  2019, Ni{\v{s}}, Serbia, June 30 - July 4, 2019, Proceedings}, volume 11545
  of {\em Lecture Notes in Computer Science}, pages 151--163. Springer, 2019.

\bibitem[Hof20a]{DBLP:conf/cocoon/Hoffmann20}
Stefan Hoffmann.
\newblock Computational complexity of synchronization under regular commutative
  constraints.
\newblock In Donghyun Kim, R.~N. Uma, Zhipeng Cai, and Dong~Hoon Lee, editors,
  {\em Computing and Combinatorics - 26th International Conference, {COCOON}
  2020, Atlanta, GA, USA, August 29-31, 2020, Proceedings}, volume 12273 of
  {\em Lecture Notes in Computer Science}, pages 460--471. Springer, 2020.

\bibitem[Hof20b]{DBLP:conf/ictcs/Hoffmann20}
Stefan Hoffmann.
\newblock On a class of constrained synchronization problems in {NP}.
\newblock In Gennaro Cordasco, Luisa Gargano, and Adele Rescigno, editors, {\em
  Proceedings of the 21th Italian Conference on Theoretical Computer Science,
  {ICTCS} 2020, Ischia, Italy}, {CEUR} Workshop Proceedings. CEUR-WS.org, 2020.

\bibitem[LT84]{DBLP:journals/eik/LatteuxT84}
Michel Latteux and Gabriel Thierrin.
\newblock On bounded context-free languages.
\newblock {\em Elektronische Informationsverarbeitung und Kybernetik (Journal
  of Information Processing and Cybernetics)}, 20(1):3--8, 1984.

\bibitem[Pin20]{Pin2020}
Jean{-}{\'{E}}ric Pin.
\newblock {\em Mathematical Foundations of Automata Theory}.
\newblock 2020.

\bibitem[SM73]{stockmeyer1973word}
Larry~J. Stockmeyer and Albert~R. Meyer.
\newblock Word problems requiring exponential time (preliminary report).
\newblock In {\em Proceedings of the fifth annual ACM Symposium on Theory of
  Computing, STOC}, pages 1--9. ACM, 1973.

\end{thebibliography}


\begin{thebibliography}{10}
\providecommand{\url}[1]{\texttt{#1}}
\providecommand{\urlprefix}{URL }
\providecommand{\doi}[1]{https://doi.org/#1}

\bibitem{adler1970similarity}
Adler, R., Weiss, B.: Similarity of Automorphisms of the Torus. American
  Mathematical Society: Memoirs of the American Mathematical Society, American
  Mathematical Society (1970)

\bibitem{Alves2020}
Alves, L.V., Pena, P.N.: Synchronism recovery of discrete event systems.
  IFAC-PapersOnLine  \textbf{53}(2),  10474--10479 (2020), 21th IFAC World
  Congress

\bibitem{DBLP:conf/icalp/AmarilliP18}
Amarilli, A., Paperman, C.: Topological sorting with regular constraints. In:
  Chatzigiannakis, I., Kaklamanis, C., Marx, D., Sannella, D. (eds.) 45th
  International Colloquium on Automata, Languages, and Programming, {ICALP}
  2018, July 9-13, 2018, Prague, Czech Republic. LIPIcs, vol.~107, pp.
  115:1--115:14. Schloss Dagstuhl - Leibniz-Zentrum f{\"{u}}r Informatik (2018)

\bibitem{Benenson2003}
Benenson, Y., Adar, R., Paz-Elizur, T., Livneh, Z., Shapiro, E.: {DNA} molecule
  provides a computing machine with both data and fuel. Proceedings of the
  National Academy of Sciences of the United States of America  \textbf{100},
  2191--2196 (2003)

\bibitem{Benenson2001}
Benenson, Y., Paz-Elizur, T., Adar, R., Keinan, E., Livneh, Z., Shapiro, E.:
  Programmable and autonomous computing machine made of biomolecules. Nature
  \textbf{414},  430--434 (2001)

\bibitem{DBLP:journals/siamcomp/BermanH77}
Berman, L., Hartmanis, J.: On isomorphisms and density of {NP} and other
  complete sets. {SIAM} J. Comput.  \textbf{6}(2),  305--322 (1977)

\bibitem{zbMATH03943051}
{Berstel}, J., {Perrin}, D.: {Theory of codes}. {Pure and Applied Mathematics,
  117. Orlando etc.: Academic Press, Inc. XIV, 433} (1985)

\bibitem{DBLP:books/daglib/0025093}
Berstel, J., Perrin, D., Reutenauer, C.: Codes and Automata, Encyclopedia of
  mathematics and its applications, vol.~129. Cambridge University Press (2010)

\bibitem{DBLP:journals/mst/BlattnerC77}
Blattner, M., Cremers, A.B.: Observations about bounded languages and
  developmental systems. Math. Syst. Theory  \textbf{10},  253--258 (1977)

\bibitem{DBLP:books/daglib/0034521}
Cassandras, C.G., Lafortune, S.: Introduction to Discrete Event Systems, Second
  Edition. Springer (2008)

\bibitem{Cer64}
{\v{C}}ern\'y, J.: Pozn\'amka k homog\'ennym experimentom s kone\v{c}n\'ymi
  automatmi. Matematicko-fyzik\'alny \v{c}asopis  \textbf{14}(3),  208--216
  (1964)

\bibitem{DBLP:journals/algorithmica/ChenI95}
Chen, Y., Ierardi, D.: The complexity of oblivious plans for orienting and
  distinguishing polygonal parts. Algorithmica  \textbf{14}(5),  367--397
  (1995)

\bibitem{DBLP:journals/et/ChoJSP93}
Cho, H., Jeong, S., Somenzi, F., Pixley, C.: Synchronizing sequences and
  symbolic traversal techniques in test generation. J. Electron. Test.
  \textbf{4}(1),  19--31 (1993)

\bibitem{DBLP:journals/dam/DassowP99}
Dassow, J., Paun, G.: On the regularity of languages generated by context-free
  evolutionary grammars. Discret. Appl. Math.  \textbf{92}(2-3),  205--209
  (1999)

\bibitem{Diekert98TR}
Diekert, V.: Makanin's algorithm for solving word equations with regular
  constraints. Report, Fakultät Informatik, Universität Stuttgart  (03 1998)

\bibitem{DBLP:journals/iandc/DiekertGH05}
Diekert, V., Guti{\'{e}}rrez, C., Hagenah, C.: The existential theory of
  equations with rational constraints in free groups is {PSPACE}-complete. Inf.
  Comput.  \textbf{202}(2),  105--140 (2005)

\bibitem{DBLP:journals/siamcomp/Eppstein90}
Eppstein, D.: Reset sequences for monotonic automata. {SIAM} J. Comput.
  \textbf{19}(3),  500--510 (1990)

\bibitem{DBLP:journals/trob/ErdmannM88}
Erdmann, M.A., Mason, M.T.: An exploration of sensorless manipulation. {IEEE}
  J. Robotics Autom.  \textbf{4}(4),  369--379 (1988)

\bibitem{DBLP:conf/mfcs/FernauGHHVW19}
Fernau, H., Gusev, V.V., Hoffmann, S., Holzer, M., Volkov, M.V., Wolf, P.:
  Computational complexity of synchronization under regular constraints. In:
  Rossmanith, P., Heggernes, P., Katoen, J. (eds.) 44th International Symposium
  on Mathematical Foundations of Computer Science, {MFCS} 2019, August 26-30,
  2019, Aachen, Germany. LIPIcs, vol.~138, pp. 63:1--63:14. Schloss Dagstuhl -
  Leibniz-Zentrum f{\"{u}}r Informatik (2019)

\bibitem{DBLP:conf/stacs/GanardiHKLM18}
Ganardi, M., Hucke, D., K{\"{o}}nig, D., Lohrey, M., Mamouras, K.: Automata
  theory on sliding windows. In: Niedermeier, R., Vall{\'{e}}e, B. (eds.) 35th
  Symposium on Theoretical Aspects of Computer Science, {STACS} 2018, February
  28 to March 3, 2018, Caen, France. LIPIcs, vol.~96, pp. 31:1--31:14. Schloss
  Dagstuhl - Leibniz-Zentrum f{\"{u}}r Informatik (2018)

\bibitem{DBLP:journals/ijfcs/GawrychowskiKRS10}
Gawrychowski, P., Krieger, D., Rampersad, N., Shallit, J.O.: Finding the growth
  rate of a regular or context-free language in polynomial time. Int. J. Found.
  Comput. Sci.  \textbf{21}(4),  597--618 (2010)

\bibitem{Ginsburg66}
Ginsburg, S.: The Mathematical Theory of Context-free Languages. McGraw-Hill
  (1966)

\bibitem{GinsburgSpanier64}
Ginsburg, S., Spanier, E.H.: Bounded {ALGOL}-like languages. Transactions of
  the American Mathematical Society  \textbf{113}(2),  333--368 (1964)

\bibitem{GinsburgSpanier66}
Ginsburg, S., Spanier, E.H.: Bounded regular sets. Proceedings of the American
  Mathematical Society  \textbf{17}(5),  1043--1049 (1966)

\bibitem{DBLP:journals/algorithmica/Goldberg93}
Goldberg, K.Y.: Orienting polygonal parts without sensors. Algorithmica
  \textbf{10}(2-4),  210--225 (1993)

\bibitem{golomb_gordon_welch_1958}
Golomb, S.W., Gordon, B., Welch, L.R.: Comma-free codes. Canadian Journal of
  Mathematics  \textbf{10},  202–209 (1958)

\bibitem{Gusev2012}
Gusev, V.V.: Synchronizing automata of bounded rank. In: Moreira, N., Reis, R.
  (eds.) Implementation and Application of Automata - 17th International
  Conference, CIAA. LNCS, vol.~7381, pp. 171--179. Springer (2012)

\bibitem{DBLP:conf/mfcs/HartmanisM80}
Hartmanis, J., Mahaney, S.R.: An essay about research on sparse {NP} complete
  sets. In: Dembinski, P. (ed.) Mathematical Foundations of Computer Science
  1980 (MFCS'80), Proceedings of the 9th Symposium, Rydzyna, Poland, September
  1-5, 1980. Lecture Notes in Computer Science, vol.~88, pp. 40--57. Springer
  (1980)

\bibitem{HerrmannKMW17}
Herrmann, A., Kutrib, M., Malcher, A., Wendlandt, M.: Descriptional complexity
  of bounded regular languages. Journal of Automata, Languages and
  Combinatorics  \textbf{22}(1-3),  93--121 (2017)

\bibitem{DBLP:conf/cocoon/Hoffmann20}
Hoffmann, S.: Computational complexity of synchronization under regular
  commutative constraints. In: Kim, D., Uma, R.N., Cai, Z., Lee, D.H. (eds.)
  Computing and Combinatorics - 26th International Conference, {COCOON} 2020,
  Atlanta, GA, USA, August 29-31, 2020, Proceedings. Lecture Notes in Computer
  Science, vol. 12273, pp. 460--471. Springer (2020)

\bibitem{DBLP:conf/ictcs/Hoffmann20}
Hoffmann, S.: On a class of constrained synchronization problems in {NP}. In:
  Cordasco, G., Gargano, L., Rescigno, A. (eds.) Proceedings of the 21th
  Italian Conference on Theoretical Computer Science, {ICTCS} 2020, Ischia,
  Italy. {CEUR} Workshop Proceedings, CEUR-WS.org (2020)

\bibitem{HopUll79}
Hopcroft, J.E., Ullman, J.D.: Introduction to Automata Theory, Languages, and
  Computation. Addison-Wesley Publishing Company (1979)

\bibitem{Hsieh1989SomeAP}
Hsieh, C., Hsu, S., Shyr, H.J.: Some algebraic properties of comma-free codes.
  Tech. rep., Kyoto University Research Information Repository (KURENAI) (1989)

\bibitem{DBLP:journals/eik/LatteuxT84}
Latteux, M., Thierrin, G.: On bounded context-free languages. Elektronische
  Informationsverarbeitung und Kybernetik (Journal of Information Processing
  and Cybernetics)  \textbf{20}(1), ~3--8 (1984)

\bibitem{Lecoutre09}
Lecoutre, C.: Constraint Networks: Techniques and Algorithms. John Wiley \&
  Sons, Ltd (2009)

\bibitem{DBLP:journals/jcss/Mahaney82}
Mahaney, S.R.: Sparse complete sets of {NP:} solution of a conjecture of
  {B}erman and {H}artmanis. J. Comput. Syst. Sci.  \textbf{25}(2),  130--143
  (1982)

\bibitem{DBLP:conf/focs/Natarajan86}
Natarajan, B.K.: An algorithmic approach to the automated design of parts
  orienters. In: 27th Annual Symposium on Foundations of Computer Science,
  Toronto, Canada, 27-29 October 1986. pp. 132--142. {IEEE} Computer Society
  (1986)

\bibitem{DBLP:journals/ijrr/Natarajan89}
Natarajan, B.K.: Some paradigms for the automated design of parts feeders. Int.
  J. Robotics Res.  \textbf{8}(6),  98--109 (1989)

\bibitem{DBLP:conf/cp/Pesant04}
Pesant, G.: A regular language membership constraint for finite sequences of
  variables. In: Wallace, M. (ed.) Principles and Practice of Constraint
  Programming - {CP} 2004, 10th International Conference, {CP} 2004, Toronto,
  Canada, September 27 - October 1, 2004, Proceedings. Lecture Notes in
  Computer Science, vol.~3258, pp. 482--495. Springer (2004)

\bibitem{Pin2020}
Pin, J.: Mathematical Foundations of Automata Theory (2020),
  \url{https://www.irif.fr/~jep/PDF/MPRI/MPRI.pdf}

\bibitem{piziak99}
Piziak, R., Odell, P.L.: Full rank factorization of matrices. Mathematics
  Magazine  \textbf{72}(3),  193--201 (1999)

\bibitem{RamadgeWonham87}
Ramadge, P.J., Wonham, W.M.: Supervisory control of a class of discrete event
  processes. SIAM Journal on Control and Optimization  \textbf{25},  206--230
  (1987)

\bibitem{DBLP:journals/ipl/Romeuf88}
Romeuf, J.: Shortest path under rational constraint. Inf. Process. Lett.
  \textbf{28}(5),  245--248 (1988)

\bibitem{DBLP:journals/tcs/Rystsov97}
Rystsov, I.: Reset words for commutative and solvable automata. Theor. Comput.
  Sci.  \textbf{172}(1-2),  273--279 (1997)

\bibitem{San2005}
Sandberg, S.: Homing and synchronizing sequences. In: Broy, M., Jonsson, B.,
  Katoen, J.P., Leucker, M., Pretschner, A. (eds.) Model-Based Testing of
  Reactive Systems. LNCS, vol.~3472, pp. 5--33. Springer (2005)

\bibitem{Trahtman09}
Trahtman, A.N.: The road coloring problem. Israel Journal of Mathematics
  \textbf{172},  51--60 (2009)

\bibitem{Vol2008}
Volkov, M.V.: Synchronizing automata and the {{\v C}}ern\'y conjecture. In:
  Mart\'{\i}n-Vide, C., Otto, F., Fernau, H. (eds.) Language and Automata
  Theory and Applications, Second International Conference, LATA. LNCS,
  vol.~5196, pp. 11--27. Springer (2008)

\bibitem{DBLP:journals/tcs/Volkov09}
Volkov, M.V.: Synchronizing automata preserving a chain of partial orders.
  Theor. Comput. Sci.  \textbf{410}(37),  3513--3519 (2009)

\bibitem{DBLP:journals/jcss/VorelR19}
Vorel, V., Roman, A.: Complexity of road coloring with prescribed reset words.
  J. Comput. Syst. Sci.  \textbf{104},  342--358 (2019)

\bibitem{wonham2019}
Wonham, W.M., Cai, K.: Supervisory Control of Discrete-Event Systems. Springer
  (2019)

\bibitem{DBLP:reference/hfl/Yu97}
Yu, S.: Regular languages. In: Rozenberg, G., Salomaa, A. (eds.) Handbook of
  Formal Languages, Volume 1: Word, Language, Grammar, pp. 41--110. Springer
  (1997)

\end{thebibliography}
\end{document}